 \documentclass[5p]{elsarticle}
\usepackage{amsmath,amssymb,amsfonts}
\usepackage{algorithmic}
\usepackage{hyperref}
\usepackage{textcomp}
\usepackage{graphicx,color}
\usepackage{mathrsfs}
\usepackage{subfigure}
\usepackage{url}
\usepackage{color}
\usepackage{dsfont}
\usepackage{bbm}
\usepackage{booktabs}
\usepackage{array}
\usepackage[table]{xcolor}
\usepackage{yfonts}

\newtheorem{theorem}{Theorem}[section]
\newtheorem{lemma}[theorem]{Lemma}
\newtheorem{definition}{Definition}
\newtheorem{corollary}[theorem]{Corollary}

\newtheorem{remark}{Remark}

\newtheorem{example}{Example}

\newcommand{\blue}[1]{{#1}}

\newcommand{\real}{\mathbb{R}}

\newcommand{\trans}{\mathsf{T}} 

\newcommand{\mc}{\mathcal}

\DeclareSymbolFont{bbold}{U}{bbold}{m}{n}
\DeclareSymbolFontAlphabet{\mathbbold}{bbold}

\newcommand\oprocendsymbol{\hbox{$\square$}}
\newcommand\oprocend{\relax\ifmmode\else\unskip\hfill\fi\oprocendsymbol}

\newcommand*{\QEDA}{\hfill\ensuremath{\blacksquare}}%

\newcommand{\QED}{\hfill \mbox{\raggedright \rule{.1in}{.1in}}}
\newenvironment{proof}{\vspace{1ex}\noindent{\itshape Proof:}\hspace{0.5em}}
{\hfill\QED\vspace{1ex}}

\begin{document}
\begin{frontmatter}
  \title{{\bf On a Security vs Privacy Trade-off in Interconnected
      Dynamical Systems}}
      \tnotetext[t1]{This material is based upon work supported in part by ARO award 71603NSYIP and in part by UCOP award LFR-18-548175.}
      \tnotetext[t2]{\emph{Email addresses:} \texttt{vkatewa@iisc.ac.in} (Vaibhav Katewa), \texttt{\{rangu003, fabiopas\}@engr.ucr.edu} (Rajasekhar Anguluri, Fabio Pasqualetti)}


\author{Vaibhav Katewa$^1$, Rajasekhar Anguluri$^2$, Fabio Pasqualetti$^2$}

\address{$^1$Department of Electrical Communication Engineering, Indian Institute of Science, Bangalore}
\address{$^2$Department of Mechanical Engineering, University of California, Riverside, CA, USA}



\begin{abstract}
  We study a security problem for interconnected systems, where each
  subsystem aims to detect local attacks using local measurements and
  information exchanged with neighboring subsystems. The subsystems
  also wish to maintain the privacy of their states and, therefore,
  use privacy mechanisms that share limited or noisy information with
  other subsystems. We quantify the privacy level based on the
  estimation error of a subsystem's state and propose a novel
  framework to compare different mechanisms based on their privacy
  guarantees. We develop a local attack detection scheme without
  assuming the knowledge of the global dynamics, which uses local and
  shared information to detect attacks with provable
  guarantees. Additionally, we quantify a trade-off between security
  and privacy of the local subsystems. Interestingly, we show that,
  for some instances of the attack, the subsystems can achieve a
  better detection performance by being more private. We provide an
  explanation for this counter-intuitive behavior and illustrate our
  results through numerical examples.
\end{abstract}

\end{frontmatter}

\section{Introduction}
Dynamical systems are becoming increasingly more distributed, diverse,
complex, and integrated with cyber components. Usually, these systems
are composed of multiple subsystems, which are interconnected among
each other via physical, cyber and other types of couplings
\cite{SMR-JPP-TKK:01}. An example of such system is the smart city,
which consists of subsystems such as the power grid, the
transportation network, the water distribution network, and
others. Although these subsystems are interconnected, it is usually
difficult to directly measure the couplings and dependencies between
them \cite{SMR-JPP-TKK:01}. As a result, they are often operated
independently without the knowledge of the other subsystems' models
and dynamics.

Modern dynamical systems are also increasingly more vulnerable to
cyber/physical attacks that can degrade their performance or may even
render them inoperable \cite{AAC-SA-SS:08b}. There have been many
recent studies on analyzing the effect of different types of attacks
on dynamical systems and possible remedial strategies (see
\cite{JG-ES-AAC-MM-MK:17} and the references therein).  A key
component of these strategies is detection of attacks using the
measurements generated by the system. Due to the autonomous nature of
the subsystems, each subsystem is primarily concerned with detection
of local attacks which affect its operation directly. However, local
attack detection capability of each subsystem is limited due to the
absence of knowledge of the dynamics and couplings with external
subsystems. One way to mutually improve the detection performance is
to share information and measurements among the subsystems. However,
these measurements may contain some confidential information about the
subsystem and, typically, subsystem operators may be willing to share
only limited information due to privacy concerns. In this paper, we
propose a privacy mechanism that limits the shared information and
characterize its privacy guarantees. Further, we develop a local
attack detection strategy using the local measurements and the limited
shared measurements from other subsystems. We also characterize the
trade-off between the detection performance and the amount/quality of
shared measurements, which reveals a counter-intuitive behavior of the involved chi-squared $(\chi^2)$ detection scheme.

\noindent \textbf{Related Work:} Centralized attack detection and
estimation schemes in dynamical systems have been studied in both
deterministic \cite{HF-PT-SD:14, FP-FD-FB:10y,YC-SK-JMFM:17} and
stochastic \cite{YM-BS:16, YC-SK-JMFM:18} settings. Recently, there
has also been studies on distributed attack detection including
information exchange among the components of a dynamical
system. Distributed strategies for attacks in power systems are
presented in \cite{HN-HI:14,SC-ZH-SK-TTK-HVP-AT:12,FD-FP-FB:11t}. In
\cite{FP-FD-FB:10y,FP-FD-FB:14}, centralized and decentralized monitor
design was presented for deterministic attack detection and
identification. In \cite{NF-GB-LC-SL-BW-BS:18, YG-XG:18}, distributed
strategies for joint attacks detection and state estimation are
presented. Residual based tests \cite{FB-AJG-GF-TP:17} and
unknown-input observer-based approaches \cite{AT-SH-KHJ:10} have also
been proposed for attack detection. A comparison between centralized
and decentralized attack detection schemes was presented in
\cite{RA-VK-FP:19}.The local detectors in \cite{RA-VK-FP:19} use only local measurements, whereas we allow the local detectors to use measurements from other subsystems as well.

Distributed fault detection techniques requiring information sharing
among the subsystems have also been widely studied. In
\cite{RMGF-TP-MMP:12, CK-MMP-TP:15, VP-MMP-CGP:15,
  XZ-QZ:12,XGY-CE:08}, fault detection for non-linear interconnected
systems is presented. These works typically use observers to estimate
the state/output, compute the residuals and compare them with
appropriate thresholds to detect faults. For linear systems,
distributed fault detection is studied using consensus-based
techniques in \cite{EF-ROS-TP-MMP:06, SS-NI-ZD-MS-KHJ:10} and
unknown-input observer-based techniques in \cite{IS-AMHT-HS-KHJ:11}.

There have also been recent studies related to privacy in dynamical
systems. Differential privacy based mechanisms in the context of
consensus, filtering and distributed optimization have been proposed
(see \cite{JC-GED-SH-JL-SM-GJP:16} and the references therein). These
works develop additive noise-based privacy mechanisms, and
characterize the trade-offs between the privacy level and the control
performance. Other privacy measures based on information theoretic
metrics like conditional entropy \cite{EA-CL-TB:15}, mutual
information \cite{FF-GN:16,TT-MS-HS-KHJ:17} and Fisher information
\cite{FF-HS:19} have also been proposed. In \cite{VK-FP-VG:16}, a
privacy vs. cooperation trade-off for multi-agent systems was
presented. In \cite{YM-RMM:17}, a privacy mechanism for consensus was
presented, where privacy is measured in terms of estimation error
covariance of the initial state. The authors in \cite{JG-AC-MK:17}
showed that the privacy mechanism can be used by an attacker to
execute stealthy attacks in a centralized setting.

In contrast to these works, we identify a novel and counter-intuitive trade-off between security and privacy in interconnected dynamical systems. In a preliminary version of this work \cite{VK-RA-FP:18a}, we compared
the detection performance between the cases when the subsystems share
full measurements (no privacy mechanism) and when they do not share
any measurements. In this paper, we introduce a privacy framework and
present an analytic characterization of privacy-performance
trade-offs.

\noindent \textbf{Contributions:} The main contributions of this paper
are as follows. First, we propose a privacy mechanism to keep the
states of a subsystem private from other subsystems in an
interconnected system. The mechanism limits both the amount and
quality of shared measurements by projecting them onto an appropriate
subspace and adding suitable noise to the measurements. This is in
contrast to prior works which use only additive noise for privacy. We
define a privacy ordering and use it to quantify and compare the
privacy of different mechanisms. Second, we propose and characterize
the performance of a chi-squared ($\chi^2$) attack detection scheme to
detect local attacks in absence of the knowledge of the global system
model. The detection scheme uses local and received measurements from
neighboring subsystems. Third, we characterize the trade-off between
the privacy level and the local detection performance \blue{in both qualitative and quantitative ways}. Interestingly,
our analysis shows that in some cases both privacy and detection
performance can be improved by sharing less information. This reveals
a counter-intuitive behavior of the widely used $\chi^2$ test for
attack detection \cite{YM-BS:16, YC-SK-JMFM:18,ASW:76}, which we
illustrate and explain.

\noindent \textbf{Mathematical notation:} $\text{Tr}(\cdot)$,
$\text{Im}(\cdot)$, $\text{Null}(\cdot)$ and\\ $\text{Rank}(\cdot)$
denote the trace, image, null space, and rank of a matrix,
respectively. $(\cdot)^{\trans}$ and $(\cdot)^{+}$ denote the
transpose and Moore-Penrose pseudo-inverse of a matrix. A positive
(semi)definite matrix $A$ is denoted by $A>0$ $(A\geq
0)$. $\text{diag}(A_1,A_2,\cdots,A_n)$ denotes a block diagonal matrix
whose block diagonal elements are $A_1,A_2,\cdots,A_n$.  The identity
matrix is denoted by $I$ (or $I_n$ to denote its dimension
explicitly). A scalar $\lambda \in\mathbb{C}$ is called a generalized
eigenvalue of $(A,B)$ if $(A-\lambda B)$ is singular. $\otimes$
denotes the Kronecker product. A zero mean Gaussian random variable
$y$ is denoted by $y\sim \mc{N}(0,\Sigma_y)$, where $\Sigma_y$ denotes
the covariance of $y$. The (central) chi-square distribution with $q$
degrees of freedom is denoted by $\chi_q^2$ and the noncentral
chi-square distribution with noncentrality parameter $\lambda$ is
denoted by $\chi_q^2(\lambda)$.  For $x\geq 0$, let $\mc{Q}_q(x)$ and
$\mc{Q}_q(x;\lambda)$ denote the right tail probabilities of a
chi-square and noncentral chi-square distributions, respectively.

\section{Problem Formulation}
We consider an interconnected discrete-time LTI dynamical system
composed of $N$ subsystems. Let $\mc{S}\triangleq \{1,2,\cdots,N\}$
denote the set of all subsystems and let
$\mc{S}_{-i}\triangleq \mc{S} \setminus \{i\}$, where $\setminus $
denotes the exclusion operator. The dynamics of the subsystems are
given by:
\begin{align} \label{eq:ss_state}
x_i(k+1) &= A_{i} x_i(k) + \blue{B_i} x_{-i}(k) + w_{i}(k), \\ \label{eq:ss_output}
y_i(k) &= C_{i}x_i(k) + v_{i}(k)  \qquad \quad i \in \mc{S},
\end{align}
where $x_{i} \in \real^{n_i}$ and $y_{i} \in \real^{p_i}$ are the
state and output/measurements of subsystem $i$, respectively. Let
$n \triangleq \sum_{i=1}^N n_i$. Subsystem $i$ is coupled with other
subsystems through the interconnection term $\blue{B_i} x_{-i}(k)$, where
$x_{-i} \triangleq [x_1^{\trans},\cdots, x_{i-1}^{\trans},
x_{i+1}^{\trans}, \cdots, x_{N}^{\trans}]^{\trans} \in \real^{n-n_i}$
denotes the states of all other subsystems. We refer to $x_{-i}$ as
the interconnection signal. Further, $w_i \in \real^{n_i}$ and $v_{i}
\in \real^{p_i}$ are the process and measurement noise,
respectively. We assume that $w_i(k) \sim \mc{N}(0,\Sigma_{w_{i}})$
and $v_i(k) \sim \mc{N}(0,\Sigma_{v_{i}})$ for all $k\geq 0$, with
$\Sigma_{w_{i}}>0$ and $\Sigma_{v_{i}}>0$. The process and measurement
noise are assumed to be white and independent for different
subsystems. Finally, we assume that the initial state  $x_i(0) \sim
\mc{N}(0,\Sigma_{x_i(0)})$ is independent of $w_i(k)$ and $v_i(k)$ for
all $k\geq 0$. We make the following assumption regarding the
interconnected system:

\noindent \emph{Assumption 1:} Subsystem $i$ has perfect knowledge of
its dynamics, i.e., it knows $(A_i, \blue{B_i}, C_i)$, the statistical
properties of $w_i$, $v_i$ and $x_i(0)$. However, it does not have
knowledge of the dynamics, states, and the statistical properties of
the noise of the other subsystems. \oprocend

\begin{remark} (\textbf{Control input}) The dynamics in
  \eqref{eq:ss_state} typically includes a control input. However,
  since each subsystem has the knowledge of its control input, its
  effect can be easily included in the attack detection
  procedure. Therefore, for the ease of presentation, we omit the
  control input.  \oprocend
\end{remark}

We consider the scenario where each subsystem can be under an
attack. We model the attacks as external linear additive inputs to the
subsystems. The dynamics of the subsystems under attack are given by
\begin{align}   \label{eq:ss_state_attack} 
x_i(k+1) &= A_{i} x_i(k) \!+\! \blue{B_i} x_{-i}(k)\! +\! \underbrace{B_i^{a} \tilde{a}_i(k)}_{\textstyle \triangleq a_i(k)} + w_{i}(k),  \\ \label{eq:ss_output_attack}
y_i(k) &= C_{i}  x_i(k) + v_{i}(k)  \qquad  i \in \mc{S},
\end{align}
where $\tilde{a}_i \in \real^{r_i}$ is the local attack input for
Subsystem $i$, which is assumed to be a deterministic but unknown
signal for all $i\in\mc{S}$. The matrix $B_i^{a}$ dictates how the
attack $\tilde{a}_i$ affects the state of Subsystem $i$, which we
assume to be unknown to Subsystem $i$.

Each subsystem is equipped with an attack monitor whose goal is to
detect the local attack using the local measurements. Since Subsystem
$i$ does not know $B_i^{a}$, it can only detect
$a_i = B_i^{a} \tilde{a}_i$. The detection procedure requires the
knowledge of the statistical properties of $y_i$ which depend on the
interconnection signal $x_{-i}$. Since the subsystems do not have
knowledge of the interconnection signals (c.f. \emph{Assumption 1}),
they share their measurements among each other to aid the local
detection of attacks (see Fig. \ref{fig:ISS}). The details of how
these shared measurements are used for attack detection are presented
in Section \ref{sec:attack_detection}.

\begin{figure}[t!]
\centering
\includegraphics[width=\columnwidth]{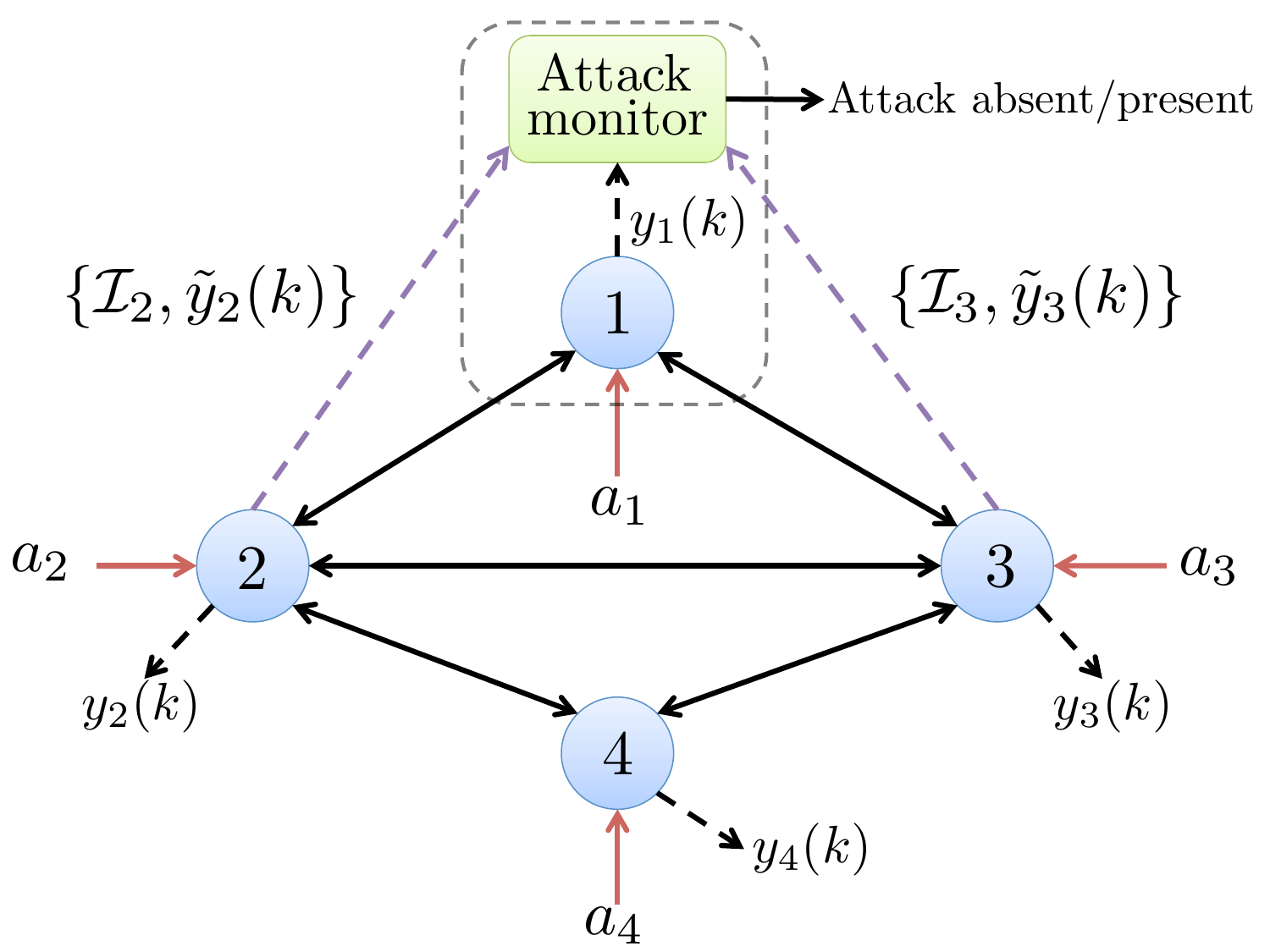}
\caption{An interconnected system consisting of $N=4$ subsystems. The
  solid lines represent state coupling among the subsystems. For
  attack detection by Subsystem $1$, its neighboring agents $2$ and
  $3$ communicate their output information to $1$ (denoted by dashed
  lines). The attack monitor associated with Subsystem $1$ uses the
  received information and the local measurements to detect attacks.}
\label{fig:ISS}
\end{figure}

While the shared measurements help in detecting local attacks, they
can reveal sensitive information of the subsystems. For instance, some
of the states/outputs of a subsystem may be confidential, which it may
not be willing to share with other subsystems. To protect the privacy
of such states/outputs, we propose a privacy mechanism $\mathcal{M}_i$
through which a subsystem limits the amount and quality of its shared
measurements.  Thus, instead of sharing the complete measurements in
\eqref{eq:ss_output_attack}, Subsystem $i$ shares limited measurements
(denoted as $\tilde{y}_{i}$) given by:
\begin{align} \label{eq:limited_meas}
  \mc{M}_i: \qquad \tilde{y}_{i}(k) &= S_i y_i(k) + \tilde{r}_i(k) \nonumber \\
                                    &= S_i C_i x_i(k) + S_i v_i(k)+
                                      \tilde{r}_{i}(k),
\end{align} 
where $S_i \in \real^{m_i \times p_i}$ is a selection matrix suitably
chosen to select a subspace of the outputs, and
$\tilde{r}_{i}(k) \sim \mc{N}(0,\Sigma_{\tilde{r}_{i}})$ is an
artificial white noise (independent of $w_i$ and $v_i$) added to
introduce additional inaccuracy in the shared measurements. Without
loss of generality, we assume $S_i$ to be full row rank for all
$i\in\mc{S}$. Thus, a subsystem can limit its shared measurement via a
combination of the following two mechanisms (i) by sharing fewer (or a
subspace of) measurements, and (ii) by sharing more noisy
measurements. Intuitively, when Subsystem $i$ limits its shared
measurements, the estimates of its states/outputs computed by the
other subsystems become more inaccurate. This prevents other
subsystems from accurately determining the confidential states/outputs
of Subsystem $i$, thereby protecting its privacy. We will explain this
phenomenon in detail in the next section.

Let the parameters corresponding to the limited measurements of
subsystem $i$ be denoted by
$\mc{I}_i \triangleq$\\
$\{C_i,S_i,\Sigma_{v_i},\Sigma_{\tilde{r}_{i}}\}$. We make the
following assumption:

\noindent \emph{Assumption 2}: Each subsystem $i\in \mc{S}$ shares its
limited measurements $\tilde{y}_i$ in \eqref{eq:limited_meas} and the
parameters $\mc{I}_i$ with all subsystems
$j\in\mc{S}_{-i}. $\footnote{To be precise, this information sharing
  is required only between \emph{neighboring} subsystems, i.e.,
  between subsystems that are directly coupled with each other in
  \eqref{eq:ss_state}.} \oprocend

Under \emph{Assumptions 1} and \emph{2}, the goal of each subsystem
$i$ is to detect the local attack $a_i$ using its local measurements
$y_i$ and the limited measurements
$\{\tilde{y}_{j}\}_{ j\in \mc{S}_{-i}}$ received from the other
subsystems (see Fig.  \ref{fig:ISS}). \blue{Further, we are interested in characterizing the trade-off between the privacy level and the detection performance.}

\section{Local attack detection} \label{sec:attack_detection}
In this section, we present the local attack detection procedure of the subsystems and characterize their detection performance. For the ease or presentation, we describe the analysis for Subsystem $1$ and remark that the procedure is analogous for the other subsystems. 

\subsection{Measurement collection}
We employ a batch detection scheme in which each subsystem collects the measurements for $k = 1,2,\cdots,T$, with $T>0$, and performs detection based on the collective measurements. 
In this subsection, we model the collected local and shared measurements for Subsystem $1$.

\noindent \textbf{Local measurements:} Let the time-aggregated local measurements, interconnection signals, attacks, process noise and measurement noise corresponding to Subsystem $1$ be respectively denoted by
\begin{align*}
y_L &\triangleq [y_1^{\trans}(1), y_1^{\trans}(2),\cdots,y_1^{\trans}(T)]^{\trans}, \\
 x &\triangleq [x_{-1}^{\trans}(0), x_{-1}^{\trans}(1),\cdots, x_{-1}^{\trans}(T-1)]^{\trans},\\
 \tilde{a} & \triangleq [\tilde{a}_1^{\trans}(0), \tilde{a}_1^{\trans}(1),\cdots,\tilde{a}_1^{\trans}(T-1)]^{\trans},\\
  w &\triangleq [w_1^{\trans}(0), w_1^{\trans}(1),\cdots,w_1^{\trans}(T-1)]^{\trans},\\
  v &\triangleq [v_1^{\trans}(1), v_1^{\trans}(2),\cdots,v_1^{\trans}(T)]^{\trans}, \quad \text{and let}\\ \addtocounter{equation}{1} \tag{\theequation} \label{eq:F_matrix}
  F(Z) &\triangleq  \begin{bmatrix}
	C_1Z & 0 & \cdots & 0\\
	C_1A_1Z & C_1Z & \cdots & 0\\
	\vdots & \vdots & \ddots & \vdots\\
	C_1A_{1}^{T-1}Z  & C_1A_{1}^{T-2}Z & \cdots & C_1Z
	\end{bmatrix}\\ &= F(I) (I_T\otimes Z).
\end{align*}

By using \eqref{eq:ss_state_attack} recursively and \eqref{eq:ss_output_attack}, the local measurements can be written as 
\begin{align} \label{eq:block_meas}
y_L &= O x_1(0) + F_x x +  F_{\tilde{a}} \tilde{a} + F_w w + v, \\ \nonumber
\text{where}\quad F_x &= F(\blue{B_1}),\: F_{\tilde{a}} = F(B_1^a), \:F_w = F(I), \quad \text{and}\\ \nonumber
	O  &\triangleq \begin{bmatrix} (C_1A_{1})^{\trans} & (C_1A_{1}^2)^{\trans} & \cdots & (C_1 A_{1}^{T})^{\trans} \end{bmatrix}^{\trans}.
\end{align} 
Note that $w \sim \mc{N}(0,\Sigma_w)$ and  $v \sim \mc{N}(0,\Sigma_v)$ with
\begin{align*}
 \Sigma_w = I_T \otimes \Sigma_{w_{1}} > 0\quad  \text{and} \quad  \Sigma_v = I_T \otimes \Sigma_{v_{1}}>0. 
 \end{align*}
Let $ v_L \triangleq O x_1(0) + F_w w + v$ denote the effective local noise in the measurement equation \eqref{eq:block_meas}. Using the fact that $(x_1(0),w,v)$ are independent, the overall local measurements of the subsystem are given by
\begin{align} \label{eq:block_meas1}
y_L &= F_x x + F_{\tilde{a}} \tilde{a} + v_L,   \quad \text{where}\\ \nonumber
v_L&\sim \mc{N}(0,\Sigma_{v_L}),\: \Sigma_{v_L} = O \Sigma_{x_1(0)} O^{\trans} + F_w \Sigma_{w} F_w^{\trans} + \Sigma_{v}>0.
\end{align}

\noindent \textbf{Shared measurements:} Let $\tilde{y}_{-1}(k) \triangleq$ \\ $[\tilde{y}_2^{\trans}(k), \tilde{y}_3^{\trans}(k),\cdots, \tilde{y}_N^{\trans}(k)]^{\trans}$ denote the limited measurements  received by Subsystem $1$ from all the other subsystems at time $k$. Further, let $v_{-1}(k)$ and $\tilde{r}_{-1}(k)$ denote similar aggregated vectors of $\left\{v_j(k) \right\}_{j\in \mc{S}_{-1}}$ and $\left\{\tilde{r}_j(k) \right\}_{j\in \mc{S}_{-1}}$, respectively.
Then, from \eqref{eq:limited_meas} we have
\begin{align} \label{eq:share_meas_1}
&\tilde{y}_{-1}(k) = S_{-1}C_{-1} x_{-1}(k) + S_{-1} v_{-1}(k) +  \tilde{r}_{-1}(k),\\ \nonumber
&\text{where} \:\:\: S_{-1}\! \triangleq \text{diag}(S_2,\cdots,S_N),  C_{-1}\! \triangleq \text{diag}(C_2,\cdots,C_N),\\ \nonumber
&v_{-1}(k)\sim \mc{N}(0,\Sigma_{v_{-1}}),\:   \Sigma_{v_{-1}} = \text{diag}(\Sigma_{v_2},\cdots,\Sigma_{v_N})>0,\\ \nonumber
&\tilde{r}_{-1}(k)\sim \mc{N}(0,\Sigma_{\tilde{r}_{-1}}),\:   \Sigma_{\tilde{r}_{-1}} = \text{diag}(\Sigma_{\tilde{r}_2},\cdots,\Sigma_{\tilde{r}_N})\geq 0.
\end{align}

Further, let the time-aggregated limited measurements received by Subsystem $1$ be denoted by $y_R  \triangleq$ \\ $[\tilde{y}_{-1}^{\trans}(0),\tilde{y}_{-1}^{\trans}(1),\cdots, \tilde{y}_{-1}^{\trans}(T-1)]^{\trans}$, and let $v_R$ denote similar time-aggregated vector of $\left\{S_{-1} v_{-1}(k) +  \tilde{r}_{-1}(k)\right\}_{k=0,\cdots, T-1}$. Then, from \eqref{eq:share_meas_1}, the overall limited measurements received by Subsystem $1$ read as 
\begin{align} \label{eq:share_meas_2}
y_R &= H x + v_R,    \quad \text{where}\\ \nonumber
H &\triangleq I_T \otimes S_{-1}C_{-1}, \quad \text{and} \quad v_R\sim \mc{N}(0,\Sigma_{v_R}) \\ \nonumber
\text{with}\quad  \Sigma_{v_R} &= I_T \otimes (S_{-1}\Sigma_{v_{-1}} S_{-1}^{\trans} + \Sigma_{\tilde{r}_{-1}})>0.
\end{align}
The goal of Subsystem $1$ is to detect the local attack using the local and received measurements given by \eqref{eq:block_meas1} and \eqref{eq:share_meas_2}, respectively. 

\subsection{Measurement processing}
Since Subsystem $1$ does not have access to the interconnection signal
$x$, it uses the received measurements to obtain an estimate of
$x$. Note that Subsystem $1$ is oblivious to the statistics of the stochastic
signal $x$. Therefore, it computes an estimate of $x$
assuming that $x$ is a deterministic but unknown quantity.

\blue{According to \eqref{eq:share_meas_2}, $y_R \sim \mc{N}(H x, \Sigma_{v_R})$, and} the Maximum Likelihood (ML) estimate of $x$ based on $y_R$
\blue{is computed by maximizing the log-likelihood function of $y_R$, and is given by:}
\begin{align}
\blue{\hat{x}} &\blue{= \arg \underset{z}{\max} \quad -\frac{1}{2} (y_R-H z)^{\trans}  \Sigma_{v_R}^{-1} (y_R-H z)} \nonumber \\ \label{eq:x_estimate}
\begin{split}
&\overset{\blue{(a)}}{=} \tilde{H}^{+} H^{\trans} \Sigma_{v_R}^{-1} y_R + (I-\tilde{H}^{+}\tilde{H}) d, \quad\text{where}\\ 
\tilde{H} &\triangleq H^{\trans} \Sigma_{v_R}^{-1} H \geq 0,
\end{split}
\end{align} 
$d$ is any real vector of appropriate dimension, \blue{and equality $(a)$ follows from Lemma \ref{lem:WLS} in the Appendix.
If $\tilde{H}$ (or equivalently $H$) is not full column rank, then
the estimate can lie anywhere in Null($\tilde{H}$) = Null($H$)
(shifted by
$\tilde{H}^{+} H^{\trans} \Sigma_{v_R}^{-1} y_R$). Thus,
the component of $x$ that lies in Null($H$) cannot be estimated
and only the component of $x$ that lies in Im($\tilde{H}$) =
Im($H^{\trans}$) can be estimated. Based on this discussion,} we
decompose $x$ as
\begin{align}
x &=  (I-\tilde{H}^{+}\tilde{H}) x + \tilde{H}^{+}\tilde{H} x \nonumber\\
&= (I-\tilde{H}^{+}\tilde{H}) x + \tilde{H}^{+} H^{\trans} \Sigma_{v_R}^{-1} H x\nonumber\\ \label{eq:x_decompose}
& \overset{\eqref{eq:share_meas_2}}{=} (I-\tilde{H}^{+}\tilde{H}) x + \tilde{H}^{+} H^{\trans} \Sigma_{v_R}^{-1}  (y_R - v_R).
\end{align}
Substituting $x$ from \eqref{eq:x_decompose} in \eqref{eq:block_meas1}, we get
\begin{align} \label{eq:block_meas2}
y_L &= F_x (I-\tilde{H}^{+}\tilde{H}) x + F_x \tilde{H}^{+} H^{\trans} \Sigma_{v_R}^{-1}  (y_R - v_R) \nonumber \\
 &+ F_{\tilde{a}} \tilde{a} + v_L.
\end{align}

Next, we process the local measurements in two steps. First, we
subtract the known term
$F_x \tilde{H}^{+} H^{\trans} \Sigma_{v_R}^{-1} y_R$. Second, we
eliminate the component $(I-\tilde{H}^{+}\tilde{H}) x$ (which cannot
be estimated) by premultiplying \eqref{eq:block_meas2} with a matrix
$M^{\trans}$, where
\begin{align} 
M &= \text{Basis of Null}\left( [F_x(I-\tilde{H}^{+}\tilde{H})]^{\trans}\right), \nonumber \\ \label{eq:left_null_matrix}
&\Rightarrow M^{\trans} F_x(I-\tilde{H}^{+}\tilde{H}) = 0.
\end{align}

Since the columns of $M$ are basis vectors, $M$ is full
column rank. The processed measurements are given by
\begin{align} 
z &= M^{\trans}(y_L -  F_x \tilde{H}^{+} H^{\trans} \Sigma_{v_R}^{-1}  y_R ) \nonumber \\ \label{eq:process_meas}
& \overset{\eqref{eq:block_meas2},\eqref{eq:left_null_matrix}}{=} M^{\trans} F_{\tilde{a}} \tilde{a} + \underbrace{M^{\trans}(v_L-F_x \tilde{H}^{+} H^{\trans} \Sigma_{v_R}^{-1}  v_R),}_{\textstyle \triangleq v_P}
\end{align}
where $v_P\sim \mc{N}(0,\Sigma_{v_P})$. The random variables $v_L$ and
$v_R$ are independent because they depend exclusively on the local and
external subsystems' noise, respectively. Using this fact
\begin{align}
\Sigma_{v_P} &=M^{\trans} \left[\Sigma_{v_L} + F_x \tilde{H}^{+} H^{\trans} \Sigma_{v_R}^{-1} \Sigma_{v_R} \Sigma_{v_R}^{-\trans} H (\tilde{H}^{+})^{\trans} F_x^{\trans}\right]M \nonumber \\ \label{eq:z_noise_var}
& \overset{\tilde{H}^{\trans} = \tilde{H}}{=} M^{\trans} \Sigma_{v_L} M + M^{\trans}F_x \tilde{H}^{+}F_x^{\trans} M\overset{(a)}{>}0,
\end{align}
where $(a)$ follows from the facts that $M$ is full column rank and
$\Sigma_{v_L}>0$. The processed measurements $z$ in
\eqref{eq:process_meas} depend only on the local attack $\tilde{a}$,
and the Gaussian noise $v_P$ whose statistics is known to Subsystem
$1$ (c.f. \emph{Assumptions 1 and 2}), i.e.
$z\sim \mc{N}(M^{\trans}F_{\tilde{a}} \tilde{a},\Sigma_{v_P})$. Thus,
Subsystem $1$ uses $z$ to perform attack detection. Note that the
attack vectors that belong to Null($M^{\trans}F_{\tilde{a}}$) cannot
be detected.

The operation of elimination of the unknown component
$(I-\tilde{H}^{+}\tilde{H}) x$ from $y_L$ also eliminates a component
of the attack $\tilde{a}$. As a result, this operation increases the
space of undetectable attack vectors from Null($F_{\tilde{a}}$) to
Null($M^{\trans}F_{\tilde{a}}$). In some cases, this operation could
also result in complete elimination of attacks as shown in the next
result.
\begin{lemma}
  Consider equation \eqref{eq:ss_state_attack} and the limited
  measurements in \eqref{eq:limited_meas}, and let $S_{-1},C_{-1},M$
  be defined in \eqref{eq:share_meas_1} and
  \eqref{eq:left_null_matrix}. If
\begin{align} \label{eq:undet_attack_lemma_cond}
\textnormal{Im}(B_1^{a}) \subseteq \textnormal{Im}\left(\blue{B_1}\left[I-(S_{-1}C_{-1})^{+}(S_{-1}C_{-1})\right]\right),
\end{align} 
then $M^{\trans}F_{\tilde{a}} = 0$.
\end{lemma}
\begin{proof}
Since Null($\tilde{H}$) = Null($H$), we have 
\begin{align*}
\tilde{H}^{+}\tilde{H} = H^{+}H = I_T \otimes (S_{-1}C_{-1})^{+}(S_{-1}C_{-1}).
\end{align*} 
Let $Z \triangleq (S_{-1}C_{-1})^{+}(S_{-1}C_{-1})$. Then, substituting $F_x$ from \eqref{eq:block_meas} in \eqref{eq:left_null_matrix}, we get
\begin{align} \label{eq:undet_attack_pf_1}
&M^{\trans} F(I) (I_T \otimes \blue{B_1})[I-I_{T}\otimes Z] = 0 \nonumber \\
\Rightarrow & M^{\trans}  F(I) (I_T  \otimes  \blue{B_1}) [I_{T} \otimes  (I -  Z)]  = 0 \nonumber \\
\Rightarrow &  M^{\trans} F(I) (I_T \otimes  \blue{B_1}[I - Z]) = 0.
\end{align}
If \eqref{eq:undet_attack_lemma_cond} holds, then there exists a matrix $P$ such that $B_1^{a} = \blue{B_1}[I-Z]P$.
Thus, from \eqref{eq:block_meas}, we have
\begin{align*}
M^{\trans} F_{\tilde{a}} &= M^{\trans}F(I) (I_T \otimes \blue{B_1}[I-Z]P )\\
& = M^{\trans}F(I) (I_T \otimes \blue{B_1}[I-Z]) (I_T \otimes P) \overset{\eqref{eq:undet_attack_pf_1}}{=} 0.
\end{align*}
\end{proof}

The above result has the following intuitive interpretation: if the
attacks lie in the subspace of the interconnections that cannot be
estimated, then eliminating these interconnections also eliminates the
attacks. In this case, the processed measurements do not have any
signature of the attacks, which, therefore, cannot be detected. This
result highlights the limitation of our measurement processing
procedure. Next, we illustrate the result using an example.

\begin{figure}[t!]
\centering
\includegraphics[scale=0.2]{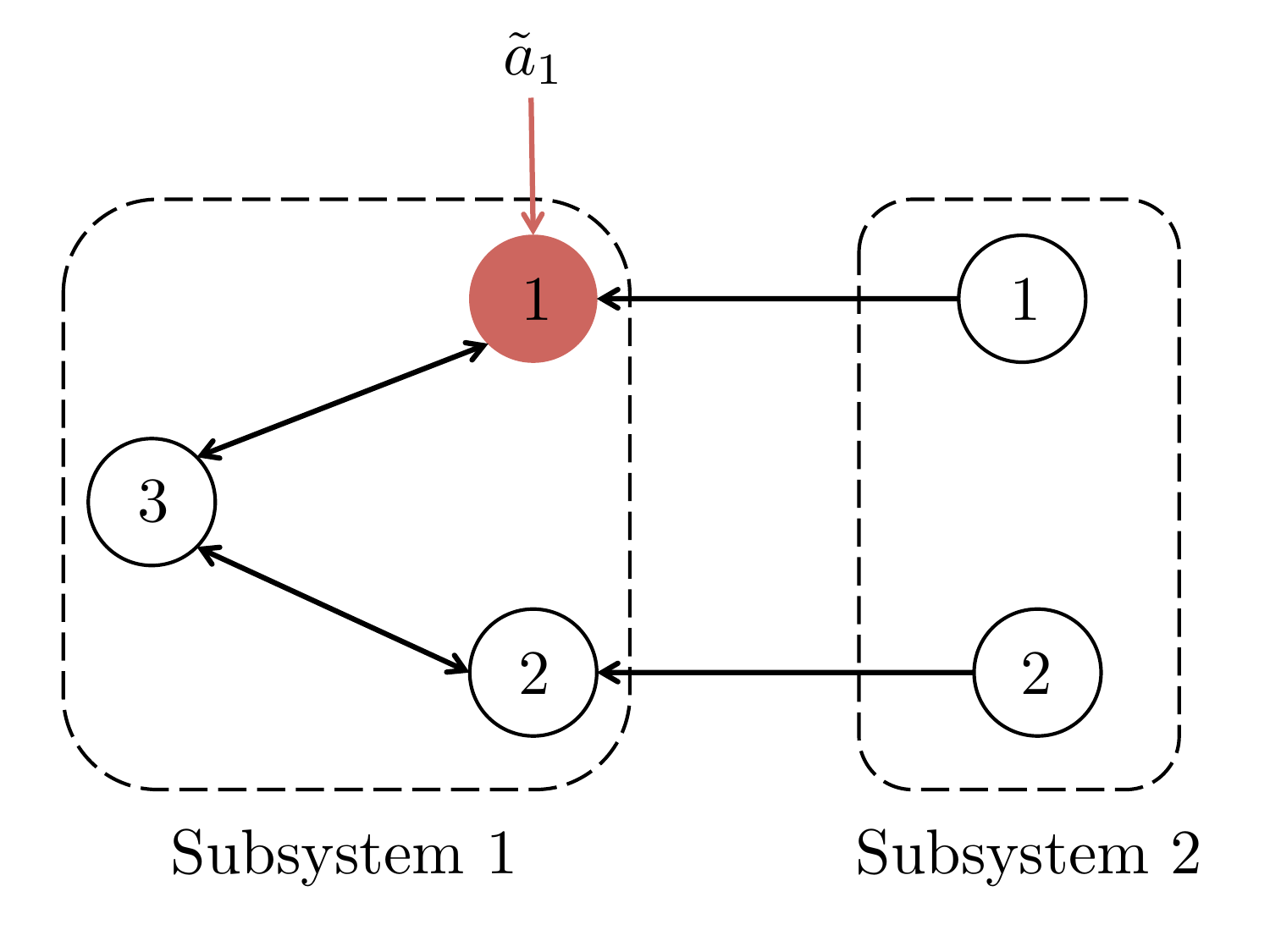}
\caption{An interconnected system consisting of two subsystems. The
  nodes denote the states of the subsystems and solid edges denote the
  couplings and interconnections of Subsystem $1$ (self edges are
  omitted). The attacked node is shaded in red.}
\label{fig:2subsys_example}
\end{figure}

\begin{example}
  Consider an interconnected subsystem consisting of two subsystems
  with the following parameters (see Fig. \ref{fig:2subsys_example}):
\begin{align*}
\small A_1 = \begin{bmatrix} 1 & 0 & -1\\ 0 & 1 & -1 \\ 1 & 1 &1 \end{bmatrix}, \:\:\blue{B_1} = \begin{bmatrix} 1 & 0\\ 0 & 1  \\ 0 & 0 \end{bmatrix},\:\: B_1^{a} = \begin{bmatrix} 1 \\ 0 \\ 0 \end{bmatrix},
\end{align*}
$C_1 = I_3, C_2 = I_2$ and $T=1$. We have $F_x = \blue{B_1}$ and $F_{\tilde{a}} = B_1^{a}$. Consider the following two cases:\\
\noindent \emph{Case (i)}: Subsystem $2$ shares its 2nd state, i.e., $S_2 = S_{-1} = \begin{bmatrix} 0 & 1 \end{bmatrix}$. In this case, Subsystem $1$ does not get information about the interconnection affecting its 1st state and the elimination of this interconnection also eliminates the attack. It can be verified that 
$M = \left[ \begin{smallmatrix}0 & 1 & 0 \\ 0 & 0 & 1 \end{smallmatrix} \right]^{\trans}$ and $M^{\trans}B_1^{a} = 0$.\\
\noindent   \emph{Case (ii)}: Subsystem $2$ shares its 1st state, i.e., $S_2 = S_{-1} = \begin{bmatrix} 1 & 0 \end{bmatrix}$. In this case, Subsystem $1$ gets information about the interconnection affecting its 1st state. Thus, its elimination is not required and this preserves the attack. It can be verified that 
$M = \left[ \begin{smallmatrix}1 & 0 & 0 \\ 0 & 0 & 1 \end{smallmatrix} \right]^{\trans}$ and $M^{\trans}B_1^{a} \neq 0$.
\end{example}  

\subsection{Statistical hypothesis testing}
The goal of Subsystem $1$ is to determine whether it is under attack
or not using the processed measurements $z$ in
\eqref{eq:process_meas}. Recall that, since Subsystem 1 does not know
$B_1^{a}$, it can only detect $a_1 = B_1^{a} \tilde{a}_1$. Let
$a \triangleq$ \\ $[ (B_1^{a}\tilde{a}_1(0))^{\trans}, \cdots,
(B_1^{a}\tilde{a}_1(T-1))^{\trans}]^{\trans}$. Then, from
\eqref{eq:block_meas}, we have $F_{\tilde{a}}\tilde{a} = F_a a$, where
$F_a = F(I)$. Thus, processed measurements are distributed
according to $z\sim \mc{N}(M^{\trans}F_{a} a,\Sigma_{v_P})$.  We cast
the attack detection problem as a binary hypothesis testing
problem. Since Subsystem $1$ does not know the attack $a$, we consider
the following \emph{composite} (simple vs. composite) testing problem
\begin{align*}
&H_0: \quad  a = 0    \quad \text{(Attack absent)} \qquad   \text{vs} \\
&H_1: \quad  a \neq 0 \quad \text{(Attack present)}
\end{align*} 
We use the Generalized Likelihood Ratio Test (GLRT) criterion
\cite{LW:04} for the above testing problem, which is given by
\begin{align} \label{eq:GLRT}
&\frac{f(z|H_0)}{\underset{a}{\sup} f(z|H_1)}\overset{H_0}{\underset{H_1}{\gtrless}} \tau' \quad \text{where,} \\ \nonumber
f(z|H_0) &= \frac{1}{\sqrt{2\pi |\Sigma_{v_P}|}} e^{-\frac{1}{2} z^{\trans} \Sigma_{v_P}^{-1} z}  \quad \text{and,} \\ \nonumber
f(z|H_1) &= \frac{1}{\sqrt{2\pi |\Sigma_{v_P}|}} e^{-\frac{1}{2} (z-M^{\trans}F_a a)^{\trans} \Sigma_{v_P}^{-1} (z-M^{\trans}F_a a)},
\end{align}
are the probability density functions of the multivariate Gaussian
distribution of $z$ under hypothesis $H_0$ and $H_1$, respectively,
and $\tau'$ is a suitable threshold. Using the result in Lemma
\ref{lem:WLS} \blue{in the Appendix} to compute the denominator in \eqref{eq:GLRT} and taking
the logarithm, the test \eqref{eq:GLRT} can be equivalently written as
\begin{align}\label{eq:GLRT_2}
t(z) &\triangleq z^{\trans} \Sigma_{v_P}^{-1} M^{\trans}F_a \tilde{M}^{+} F_a^{\trans} M \Sigma_{v_P}^{-1} z \overset{H_1}{\underset{H_0}{\gtrless}} \tau, \\ \nonumber
 \text{where}& \quad \tilde{M} = F_a^{\trans} M \Sigma_{v_P}^{-1} M^{\trans}F_a,
\end{align}
and $\tau\geq 0$ is the threshold. The above test is a $\chi^2$ test
since the test statistics $t(z)$ follows a chi-squared distribution
(see Lemma \ref{lem:ts_distribution}). The next result simplifies the
test statistics $t(z)$ and provides an interpretation of the test.
\begin{lemma} \label{lem:test_stat_simpl} \textbf{(Simplification of test statistics)}
Let $\Sigma_{v_P}^{-1} = R^{\trans} R$ denote the Cholesky decomposition of $\Sigma_{v_P}^{-1}$. Then,
\begin{align} \label{eq:test_stat_simpl}
\Sigma_{v_P}^{-1} M^{\trans}F_a \tilde{M}^{+} F_a^{\trans} M \Sigma_{v_P}^{-1} = R^{\trans} U U^{\trans} R,
\end{align}
where $U$ is a matrix whose columns are the orthonormal basis vectors of $\textnormal{Im}(RM^{\trans}F_a)$.
\end{lemma}    
\begin{proof}
Let $M_1 \triangleq M^{\trans} F_a$. Then
\begin{align*}
\tilde{M}^{+} = (M_1^{\trans} R^{\trans} R M_1)^{+} = (R M_1)^{+}  ((R M_1)^{+})^{\trans}.
\end{align*}
Thus, we have
\begin{align*}
\Sigma_{v_P}^{-1} & M^{\trans}F_a \tilde{M}^{+} F_a^{\trans} M \Sigma_{v_P}^{-1} \\
 &= (R^{\trans} R) M_1 (R M_1)^{+} ((R M_1)^{+})^{\trans} M_1^{\trans} (R^{\trans} R) \\
 &= R^{\trans} (R M_1) (R M_1)^{+}  (RM_1) (R M_1)^{+} R \\
 &= R^{\trans} (R M_1) (R M_1)^{+} R.
\end{align*} 
Since $R M_1 (R M_1)^{+}$ is the orthogonal projection operator on $\text{Im}(R M_1)$, $ R M_1 (R M_1)^{+} = U U^{\trans}$, and the proof is complete.
\end{proof}

Using Lemma \ref{lem:test_stat_simpl}, the test \eqref{eq:GLRT_2} can be written as 
\begin{align} \label{eq:GLRT_3}
t(z) = z^{\trans} R^{\trans} U U^{\trans} R z \overset{H_1}{\underset{H_0}{\gtrless}} \tau .
\end{align}
Thus, the test compares the energy of the signal $U^{\trans}Rz$ with a
given threshold to detect the attacks.  Next, we derive the
distribution of the test statistics under both hypothesis.
\begin{lemma} \label{lem:ts_distribution} \textbf{(Distribution of test statistics)}
The distribution of test statistics $t(z)$ in \eqref{eq:GLRT_3} is given by
\begin{align} \label{eq:dist_H0}
 t(z) &\sim \chi_q^2  \quad \text{under $H_0$},  \\ \label{eq:dist_H1}
 t(z) & \sim \chi_q^2(\lambda \triangleq a^{\trans} \Lambda a) \quad \text{under $H_1$},
 \end{align}
 where $q = \textnormal{Rank}(M^{\trans}F_a)$ and $\Lambda = F_a^{\trans}M \Sigma_{v_P}^{-1} M^{\trans} F_a$.
\end{lemma} 
\begin{proof}
By the definition of $U$ in \eqref{eq:test_stat_simpl}, and recalling $\Sigma_{v_P}^{-1} = R^{\trans} R$ with $R$ being non-singular, we have
\begin{align*}
\text{Rank}(U^{\trans}U)=\text{Rank}(U) = \text{Rank}(RM^{\trans}F_a) = \text{Rank}(M^{\trans}F_a).
\end{align*}
Let $z' = U^{\trans} R z$. Under $H_0$,  $z\sim \mc{N}(0,\Sigma_{v_P})$. Thus, 
\begin{align*}
z'\sim \mc{N}(0,U^{\trans} R\Sigma_{v_P} R^{\trans} U) \overset{(a)}{=} \mc{N}(0,I_q),
\end{align*}
where $(a)$ follows from $R\Sigma_{v_P} R^{\trans} = I$ and $U^{\trans}U = I_q$. Therefore, $t(z) = (z')^{\trans} z'\sim \chi_q^2$.

Let $M_1 = M^{\trans} F_a$. Under $H_1$, $z\sim \mc{N}(M_1 a,\Sigma_{v_P})$. Thus, 
\begin{align*}
z'&\sim \mc{N}(U^{\trans} R M_1 a, I_q)\\
\Rightarrow t(z) = (z')^{\trans} z' &\sim \chi_q^2(a^{\trans} M_1^{\trans} R^{T} U U^{\trans} R M_1 a).
\end{align*}
Using $UU^{\trans} = R M_1 (R M_1)^{+}$ from the proof of Lemma \ref{lem:test_stat_simpl}, we have
\begin{align*}
&a^{\trans} M_1^{\trans} R^{T} U U^{\trans} R M_1 a = a^{\trans} (R M_1)^{\trans} (R M_1) (R M_1)^{+} (R M_1) a \\
&= a^{\trans} (R M_1)^{\trans} (R M_1) a = a^{\trans} M_1^{\trans} \Sigma_{v_P}^{-1} M_1 a = \lambda,
\end{align*}
and the proof is complete.
\end{proof}

\begin{remark} \textbf{(Interpretation of detection parameters
    $(q,\lambda)$)} The parameter $q$ denotes the number of
  independent observations of the attack vector $a$ in the processed
  measurements \eqref{eq:process_meas}. The parameter $\lambda$ can be
  interpreted as the signal to noise ratio (SNR) of the processed
  measurements in \eqref{eq:process_meas}, where the signal of
  interest is the attack. \oprocend
\end{remark}

Next, we characterize the performance of the test
\eqref{eq:GLRT_2}. Let the probability of false alarm and probability
of detection for the test be respectively denoted by
\begin{align*}
P_F &= \text{Prob}(t(z)> \tau | H_0)  \overset{(a)}{=} \mc{Q}_q(\tau) \quad \text{and,}\\
P_D &= \text{Prob}(t(z)> \tau | H_1) \overset{(b)}{=} \mc{Q}_q(\tau;\lambda),
\end{align*}
where $(a)$ and $(b)$ follow from \eqref{eq:dist_H0} and
\eqref{eq:dist_H1}, respectively. Recall that $\mc{Q}_q(x)$ and
$\mc{Q}_q(x;\lambda)$ denote the right tail probabilities of
chi-square and noncentral chi-square distributions,
respectively. Inspired by the Neyman-Pearson test framework, we select
the size ($P_F$) of the test and determine the threshold $\tau$ which
provides the desired size. Then, we use the threshold to perform the
test and compute the detection probability. Thus, we have
\begin{align}
\tau(q,P_F) &= \mc{Q}_q^{-1}(P_F), \\ \label{eq:Pd}
P_D(q,\lambda,P_F) &= \mc{Q}_q(\tau(q,P_F);\lambda).
\end{align}
The arguments in $\tau(q,P_F)$ and $P_D(q,\lambda,P_F)$ explicitly
denote the dependence of these quantities on the detection parameters
$(q,\lambda)$ and the probability of false alarm ($P_F$). Note that
the detection performance of Subsystem $1$ is characterized by the
pair $(P_F, P_D)$, where a lower value of $P_F$ and a higher value of
$P_D$ is desirable. Later, in order to compare the performance of two
different tests, we select a common value of $P_F$ for both of them,
and then compare the detection probability $P_D$.

The next result states the dependence of the detection probability on
the detection parameters $(q,\lambda)$.

\begin{lemma} \label{lem:Pd_dependence} \textbf{(Dependence of
    detection performance on detection parameters $(q,\lambda)$)} For
  any given false alarm probability $P_F$, the detection probability
  $P_D(q,\lambda,P_F)$ is decreasing in $q$ and increasing in
  $\lambda$.
\end{lemma}
\begin{proof}
  Since $P_F$ is fixed, we omit it in the notation. It is a standard
  result that for a fixed $q$ (and $\tau(q)$), the CDF
  $(=1-\mc{Q}_q(\tau(q);\lambda) = 1-P_D(q,\lambda))$ of a noncentral
  chi-square random variable is decreasing in $\lambda$
  \cite{NLJ-SK-NB:95}. Thus, $P_D(q,\lambda)$ is increasing in
  $\lambda$.

  Next, we have \cite{NLJ-SK-NB:95}
  \begin{align*}
    P_D(q,\lambda)  = e^{-\lambda /2} \sum_{j=0}^{\infty} \frac{(\lambda/2)^{j}}{j!} \mc{Q}_{q + 2j}(\tau(q)).
  \end{align*} 
  From \cite[Corollary 3.1]{EF-RZ:08}, it follows that
  $\mc{Q}_{q + 2j}(\tau(q)) = \mc{Q}_{q + 2j}(\mc{Q}_q^{-1}(P_F))$ is
  decreasing in $q$ for all $j>0$. Thus, $P_D(q,\lambda)$ is
  decreasing in $q$.
\end{proof}

Figure \ref{fig:Pd_vs_det_param} illustrates the dependence of the
detection probability on the parameters $(q,\lambda)$. Lemma
\ref{lem:Pd_dependence} implies that for a fixed $q$, a higher SNR
($\lambda$) leads to a better detection performance, which is
intuitive. However, for a fixed $\lambda$, an increase in the number
of independent observations ($q$) results in degradation of the
detection performance. This counter-intuitive behavior is due to the
fact that the GLRT in \eqref{eq:GLRT} is not an uniformly most
powerful (UMP) test for all values of the attack $a$. In fact, a UMP
test does not exist in this case \cite{ELL-JPR:05}. Thus, the test can
perform better for some particular attack values while it may not
perform as good for other attack values. This suboptimality is an
inherent property of the GLRT in \eqref{eq:GLRT}. It arises due to the
composite nature of the test and the fact that the value of the attack
vector $a$ is not known to the attack monitor.

\begin{figure}[t]
  \centering
  \subfigure[]{
  \includegraphics[width=\columnwidth]{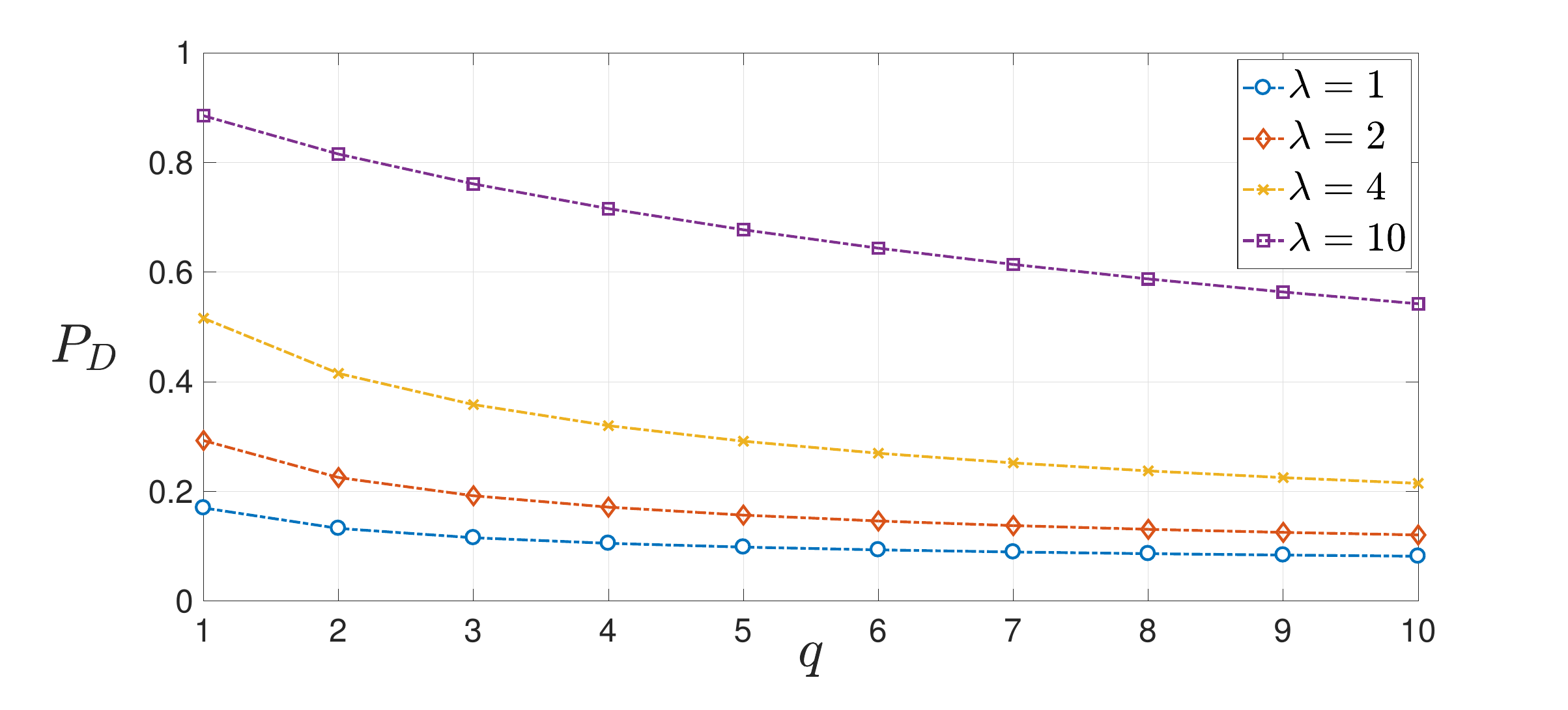} \label{fig:Pd_vs_q}} 
  \subfigure[]{
  \includegraphics[width=\columnwidth]{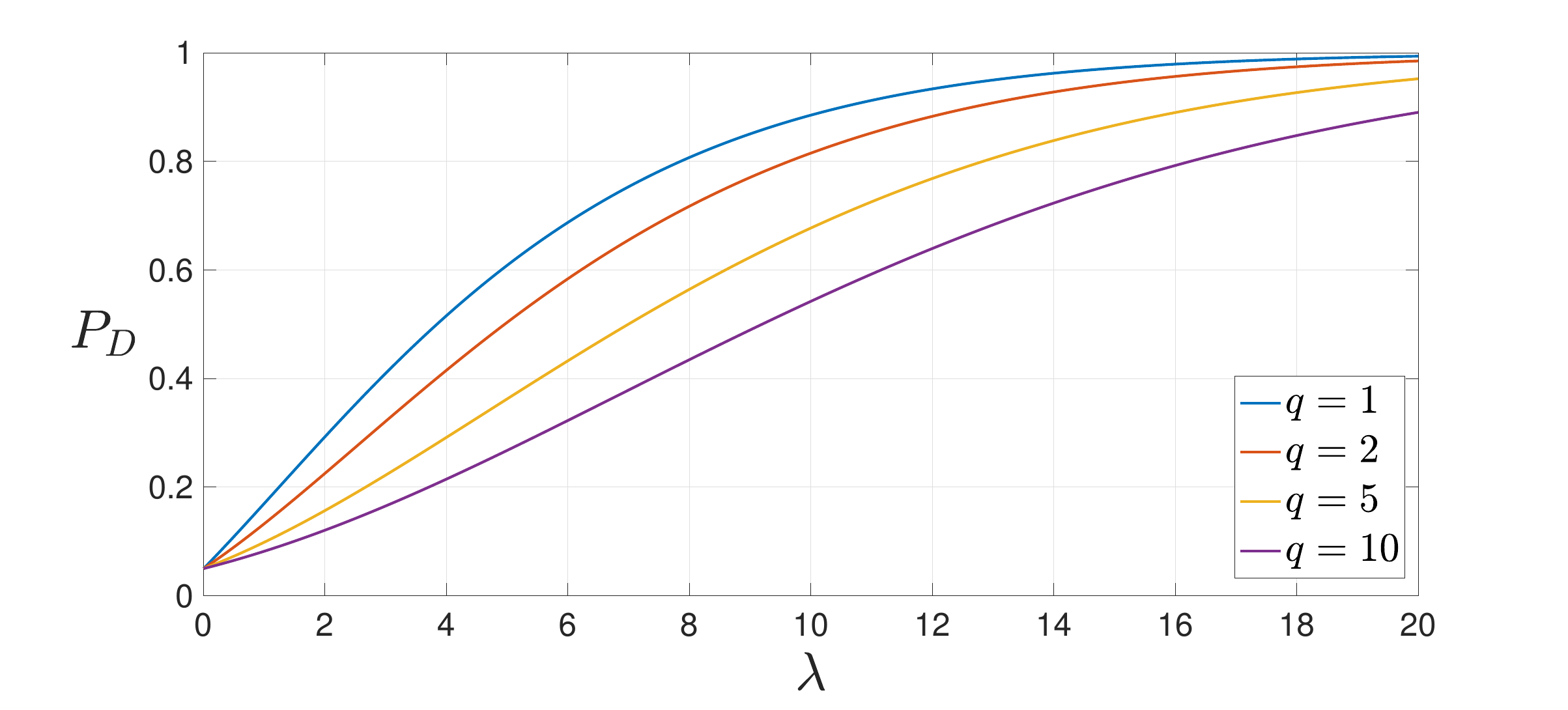} \label{fig:Pd_vs_lambda}}
\caption{Dependence of the detection probability $P_D$ on the
  detection parameters $(q,\lambda)$ for a fixed $P_F = 0.05$. $P_D$
  decreases monotonically with $q$ (subfigure (a)), whereas it
  increases monotonically with $\lambda$ (subfigure (b)).}
  \label{fig:Pd_vs_det_param}
\end{figure}

\begin{remark} \textbf{(Composite vs. simple
    test)} \label{rem:comp_vs_sim_test} If the value of the attack
  vector (say $a_1$) is known, we can cast a \emph{simple} (simple
  vs. simple) binary hypothesis testing problem as $H_0: a=0$
  vs. $H_1:a=a_1$ and use the standard Likelihood Ratio Test
  criterion for detection. In this case the detection probability
  depends only on $P_F$ and SNR $(\lambda)$, and for any given $P_F$,
  the detection performance improves as the SNR increases.  \oprocend
\end{remark}

\section{Privacy quantification}
In this section, we quantify the privacy of the mechanism $\mc{M}_i$
in \eqref{eq:limited_meas} in terms of the estimation error covariance
of the state $x_i$. \blue{For simplicity, we assume $i \neq 1$, and
  this estimation is performed by Subsystem $1$, which is directly
  coupled with Subsystem $i$ and receives limited measurements from
  it.} Then, we use this quantification to compare and rank different
privacy mechanisms.

We use a batch estimation scheme in which the estimate is computed
based on the collective measurements obtained for $k=0, 1,\cdots, T-1$,
with $T>0$ . Let
$\tilde{y}_i = [\tilde{y}_i^{\trans}(0),\cdots,
\tilde{y}_i^{\trans}(T-1)]^{\trans}$, and let $x_i$, $v_i$,
$\tilde{r}_i$ be similar time-aggregated vectors of
$x_i(k), v_i(k), \tilde{r}_i(k)$, respectively. Then, using
\eqref{eq:limited_meas}, we have
\begin{align} \label{eq:limited_meas_tag}
\tilde{y}_i = (\underbrace{I_T \otimes S_iC_i}_{\textstyle \triangleq H_i}) x_i + \underbrace{(I_T \otimes S_i)v_i + \tilde{r}_i}_{\textstyle \triangleq r_i}, 
\end{align}
where $r_i \sim \mc{N}(0, \Sigma_{r_i})$ with
$\Sigma_{r_i} = I_T \otimes (S_i \Sigma_{v_i} S_i^{\trans} +
\Sigma_{\tilde{r}_i})$. Note that Subsystem $1$ that receives
measurements \eqref{eq:limited_meas_tag} from Subsystem $i$ knows
$\{H_i,\Sigma_{r_i}\}$ (c.f. \emph{Assumption} 2). However, it is
oblivious to the statistics of the confidential stochastic signal
$x_i$. Therefore, it computes an estimate of $x_i$ assuming that $x_i$
is a deterministic but unknown quantity. \blue{Further, this estimate
  is computed by Subsystem $1$ using the measurements received only
  from Subsystem $i$, and it does not use its local measurements or
  the measurements received from other subsystems for this
  purpose. The reason is twofold. First, although the local
  measurements $y_L$ of Subsystem $1$ depend on $x_i$ due to the
  interconnected nature of the system (see \eqref{eq:block_meas1}),
  they cannot be used due to the presence of the unknown attack on
  Subsystem $1$ given by $F_{\tilde{a}}\tilde{a} = F_a a$, where
  $F_a=F(I)$ is known and $a$ is unknown. If we try to eliminate these
  unknown attacks by pre-multiplying \eqref{eq:block_meas1} with a
  matrix $M'$, where $(M')^{\trans}$ is the basis of Null($F_a$), this
  operation also eliminates $x$ (and $x_i$), since
  $\text{Null}(F_a)\subseteq \text{Null}(F_x)$. Second, the
  measurements received from other subsystems cannot be used since
  Subsystem $1$ does not have the knowledge of the dynamics or attacks
  on these other subsystems.}\footnote{\blue{Due to these reasons, the
    estimation capability of any Subsystem $j\in\mc{S}_{-i}$ trying to
    infer $x_i$ will be the same.}}

 \blue{According to \eqref{eq:limited_meas_tag}, $\tilde{y}_i \sim \mc{N}(H_ix_i, \Sigma_{r_i})$, and} the Maximum Likelihood (ML) estimate of $x_i$ based on $\tilde{y}_i$
\blue{is computed by maximizing the log-likelihood function of $\tilde{y}_i$, and is given by:}
\begin{align}
\blue{\hat{x}_i} &\blue{= \arg \underset{z}{\max} \quad -\frac{1}{2} (\tilde{y}_i-H_i z)^{\trans}  \Sigma_{r_i}^{-1} (\tilde{y}_i-H_i z)} \nonumber \\ \label{eq:x_i_estimate}
\begin{split}
&\overset{\blue{(a)}}{=} \tilde{H}_i^{+} H_i^{\trans} \Sigma_{r_i}^{-1} \tilde{y}_i + (I-\tilde{H}_i^{+}\tilde{H}_i) d_i, \quad\text{where}\\ 
\tilde{H}_i &\triangleq H_i^{\trans} \Sigma_{r_i}^{-1} H_i \geq 0,
\end{split}
\end{align} 
$d_i$ is any real vector of appropriate dimension, \blue{and equality $(a)$ follows from Lemma \ref{lem:WLS} in the Appendix}. If
$\tilde{H}_i$ (or equivalently $H_i$) is not full column rank, then
the estimate can lie anywhere in Null($\tilde{H}_i$) = Null($H_i$)
(shifted by
$\tilde{H}_i^{+} H_i^{\trans} \Sigma_{r_i}^{-1} \tilde{y}_i$). Thus,
the component of $x_i$ that lies in Null($H_i$) cannot be estimated
and only the component of $x_i$ that lies in Im($\tilde{H}_i$) =
Im($H_i^{\trans}$) can be estimated. Let
$\mc{P}_i \triangleq \tilde{H}_i^{+}\tilde{H}_i$ denote the projection
operator on Im($\tilde{H}_i$). The estimation error in this subspace
is given by:
\begin{align}
e_i &= \mc{P}_i x_i - \mc{P}_i \hat{x}_i  = \tilde{H}_i^{+}\tilde{H}_i x_i - \tilde{H}_i^{+} H_i^{\trans} \Sigma_{r_i}^{-1} \tilde{y}_i \nonumber \\
& = - \tilde{H}_i^{+} H_i^{\trans} \Sigma_{r_i}^{-1} r_i,
\end{align}
and the estimation error covariance is given by:
\begin{align}
\Sigma_{e_i} &= \mathbb{E}[\tilde{H}_i^{+} H_i^{\trans} \Sigma_{r_i}^{-1} r_i r_i^{\trans}  \Sigma_{r_i}^{-1} H_i \tilde{H}_i^{+}] \nonumber \\ \label{eq:err_cov}
& \blue{=\tilde{H}_i^{+}\underbrace{ H_i^{\trans} \Sigma_{r_i}^{-1} H_i}_{\tilde{H}_i} \tilde{H}_i^{+}}  =  \tilde{H}_i^{+}.
\end{align}
Note that since the model in \eqref{eq:limited_meas_tag} is linear
with Gaussian noise, $\mc{P}_i \hat{x}_i$ is the minimum-variance
unbiased (MVU) estimate of $x_i$ projected on
Im($H_i^{\trans}$). Thus, the covariance $\Sigma_{e_i}$ captures the
fundamental limit on how accurately $\mc{P}_i x_i$ can be estimated
and, therefore, it is a suitable metric to quantify privacy.

The privacy level of mechanism $\mc{M}_i$ in \eqref{eq:limited_meas}
is characterized by two quantities: (i) rank($S_i$), and (ii)
$\Sigma_{e_i}$. Intuitively, if rank($S_i$) is small, then Subsystem
$i$ shares fewer measurements and, as a result, the component of $x_i$
that cannot be estimated $((I-\tilde{H}_i^{+}\tilde{H}_i) x_i)$
becomes large. Further, if $\Sigma_{e_i}$ is large (in a positive
semi-definite sense), this implies that the estimation accuracy of the
component of $x_i$ that can be estimated
$(\tilde{H}_i^{+}\tilde{H}_ix_i)$ is worse. Thus, a lower value of
rank($S_i$) and a larger value of $\Sigma_{e_i}$ implies a larger
level of privacy. Based on this discussion, we next define an ordering
between two privacy mechanisms.

Consider two privacy mechanisms $\mc{M}_i^{(1)}$ and $\mc{M}_i^{(2)}$,
and let $\tilde{y}_i^{(k)}, \hat{x}_i^{(k)}$, $k = 1,2$ denote the
limited measurements and estimates corresponding to the two
mechanisms, respectively. Further, let
$S_i^{(k)},H_i^{(k)},\tilde{H}_i^{(k)},\mc{P}_i^{(k)},\Sigma_{e_i}^{(k)},$
$k = 1,2$ denote the quantities defined above corresponding to
$\mc{M}_i^{(1)}$ and $\mc{M}_i^{(2)}$.
\begin{definition}\label{def:priv_order} \textbf{(Privacy ordering)}
  Mechanism $\mc{M}_i^{(2)}$ is more private than $\mc{M}_i^{(1)}$,
  denoted by $\mc{M}_i^{(2)}\geq \mc{M}_i^{(1)}$, if
\begin{align} \label{eq:priv_order}
\begin{split}
&(i) \:\textnormal{Im}\left((S_i^{(2)})^{\trans}\right) \subseteq \textnormal{Im}\left((S_i^{(1)})^{\trans}\right) \quad \text{and,} \hspace{50pt}\\ 
&(ii)\: \Sigma_{e_i}^{(2)} \geq \mc{P}_i^{(2)}\Sigma_{e_i}^{(1)} \mc{P}_i^{(2)}. \hspace*{105pt} \oprocend
\end{split}
\end{align}
\end{definition}

The first condition implies that $\tilde{y}_i^{(2)}$ is a limited
version of $\tilde{y}_i^{(1)}$ and is required for the ordering to be
well defined. Under this condition, it is easy to see that
$\text{Im}(H_i^{(2)}) = \text{Im}(\mc{P}_i^{(2)}) \subseteq
\text{Im}(H_i^{(1)}) = \text{Im}(\mc{P}_i^{(1)})$. Thus, the estimated
component $\mc{P}_i^{(2)} \hat{x}_i^{(2)}$ lies in a subspace that is
contained in the subspace of the estimated component
$\mc{P}_i^{(1)} \hat{x}_i^{(1)}$. For a fair comparison between the
two mechanisms, we consider the projection of
$\mc{P}_i^{(1)} \hat{x}_i^{(1)}$ on $\text{Im}(\mc{P}_i^{(2)})$, given
by\\
$\mc{P}_i^{(2)} \mc{P}_i^{(1)} \hat{x}_i^{(1)} = \mc{P}_i^{(2)}
\hat{x}_i^{(1)}$. Then, we compare its estimation error (given by
$\mc{P}_i^{(2)}\Sigma_{e_i}^{(1)} \mc{P}_i^{(2)}$) with the estimation
error of $\mc{P}_i^{(2)} \hat{x}_i^{(2)}$ (given by
$\Sigma_{e_i}^{(2)}$) to obtain the second condition in
\eqref{eq:priv_order}. Next, we present an example to illustrate
Definition \ref{def:priv_order}.

\begin{example}
  Let $x_i \in \mathbb{R}^{2}$, $C_i = I_2$, $T = 1$, and consider two
  privacy mechanisms given by:
\begin{align*}
&\mc{M}_i^{(1)}: \qquad \qquad \tilde{y}_i^{(1)} = (x_i + v_i) + \tilde{r}_i^{(1)},\\
&\mc{M}_i^{(2)}: \qquad \qquad \tilde{y}_i^{(2)} = \begin{bmatrix} 1 & 0 \end{bmatrix}(x_i + v_i) + \tilde{r}_i^{(2)},
\end{align*}
with $\Sigma_{v_i} = \Sigma_{\tilde{r}_i}^{(1)} = I_2$ and
$\Sigma_{\tilde{r}_i}^{(2)} = \alpha \geq 0$. Mechanism
$\mc{M}_i^{(1)}$ shares both components of the measurement vector
$y_i$ ($S_i^{(1)} = I_2$) whereas $\mc{M}_i^{(2)}$ shares only the
first component ($S_i^{(2)} = [1 \:\:0]$), and both add some
artificial noise. The state estimates under the two mechanisms (using
\eqref{eq:x_i_estimate}) are given by
\begin{align*}
 \hat{x}_i^{(1)} = \tilde{y}_i^{(1)} \quad \text{and} \quad \hat{x}_i^{(2)} = \left[\begin{matrix} 1\\ 0 \end{matrix}\right] \tilde{y}_i^{(2)} + \left[\begin{matrix} 0 & 0\\ 0 & 1\end{matrix}\right] d_i. 
 \end{align*}
 Thus, under $\mc{M}_i^{(1)}$ both components of $x_i$ can be
 estimated while under $\mc{M}_i^{(2)}$, only the first component can
 be estimated. Further, we have
 $\Sigma_{e_i}^{(1)} = 2I_2,\Sigma_{e_i}^{(2)} =
 \left[\begin{smallmatrix} 1+\alpha & 0\\ 0 &
     0\end{smallmatrix}\right]$ and
 $\mc{P}_i^{(2)} = \left[\begin{smallmatrix} 1 & 0\\ 0 &
     0\end{smallmatrix}\right]$. Thus, the estimation error covariance
 of the first component of $x_i$ under $\mc{M}_i^{(1)}$ and
 $\mc{M}_i^{(2)}$ are $2$ and $1+\alpha$, respectively, and
 $\mc{M}_i^{(2)}$ is more private than $\mc{M}_i^{(1)}$ if
 $\alpha\geq 1$.

 On the other hand, if $\alpha < 1$, then an ordering between the
 mechanisms cannot be established. In this case, under
 $\mc{M}_i^{(1)}$, both the state components can be estimated but the
 estimation error in first component is large. In contrast, under
 $\mc{M}_i^{(2)}$, only the first component can be estimated but its
 estimation error is small.  \oprocend
\end{example}

Next, we state a sufficient condition on the noise added by two
privacy mechanisms that guarantee the ordering of the mechanisms. This
condition implies that, if one privacy mechanism shares a subspace of
the measurements of the other mechanism and injects a sufficiently
large amount of noise, then it is more private.

\begin{lemma} (\textbf{Sufficient condition for privacy
    ordering}) \label{eq:suff_cond_priv_order} Consider two privacy
  mechanisms $\mc{M}_i^{(1)}$ and $\mc{M}_i^{(2)}$ in
  \eqref{eq:limited_meas} with parameters
  $(S_i^{(k)}, \Sigma_{\tilde{r}_{i}}^{(k)})$, $k=1,2$ that satisfy
  condition $(i)$ of \eqref{eq:priv_order}. Let $P$ be a full row rank matrix that satisfies $S_i^{(2)} = P S_i^{(1)}$. If
\begin{align} \label{eq:noise_relation}
\Sigma_{\tilde{r}_{i}}^{(2)} \geq P \Sigma_{\tilde{r}_{i}}^{(1)} P^{\trans},
\end{align}
then $\mc{M}_i^{(2)}$ is more private than $\mc{M}_i^{(1)}$.
\end{lemma}
\begin{proof}
From \eqref{eq:limited_meas_tag} and \eqref{eq:x_i_estimate}, we have
\begin{align*}
\tilde{H}_i^{(k)} = I_T\otimes \underbrace{(S_i^{(k)}C_i)^{\trans} \left[S_i^{(k)} \Sigma_{v_i}(S_i^{(k)})^{\trans} + \Sigma_{\tilde{r}_{i}}^{(k)}\right]^{-1} S_i^{(k)}C_i}_{\triangleq Y^{(k)}}.
\end{align*}
Since $(S_i^{(1)}, S_i^{(2)})$ satisfy \eqref{eq:priv_order} $(i)$, there always exist a full row rank matrix $P$ satisfying $S_i^{(2)} = P S_i^{(1)}$. Next we have,
\begin{align} \label{eq:noise_reln_pf1}
&Y^{(2)}\!\! = \!(S_i^{(1)}C_i)^{\trans} P^{\trans} \!\left[P S_i^{(1)} \Sigma_{v_i}(S_i^{(1)})^{\trans} P^{\trans} \!+\! \Sigma_{\tilde{r}_{i}}^{(2)}\right]^{-1}\!\! P S_i^{(1)}C_i  \nonumber \\
 &=  \!(S_i^{(1)}C_i)^{\trans}\! P^{\trans} \!\left[\!P (S_i^{(1)} \Sigma_{v_i}(S_i^{(1)})^{\trans}\! \!+\! \Sigma_{\tilde{r}_{i}}^{(1)}) P^{\trans}\! \!+\! E \right]^{-1}\!\! \!P S_i^{(1)}C_i \nonumber \\
 & \overset{(a)}{\leq}\!\! (S_i^{(1)}C_i)^{\trans}\! \! \left[S_i^{(1)} \Sigma_{v_i}(S_i^{(1)})^{\trans}\! \!+\! \Sigma_{\tilde{r}_{i}}^{(1)}\right]^{-1} \!\!\! S_i^{(1)}C_i = Y^{(1)}
\end{align}
where $E \triangleq \Sigma_{\tilde{r}_{i}}^{(2)} - P \Sigma_{\tilde{r}_{i}}^{(1)} P^{\trans}$ and $(a)$ follows from $E\geq 0$ (using \eqref{eq:noise_relation}) and Lemma \ref{lem:sigma_ordering} \blue{in the Appendix}. From \eqref{eq:noise_reln_pf1}, it follows that
\begin{align*}
 &\tilde{H}_i^{(2)} \leq \tilde{H}_i^{(1)} \overset{(b)}{\Rightarrow}  \tilde{H}_i^{(2)} \geq \tilde{H}_i^{(2)}  (\tilde{H}_i^{(1)})^{+} \tilde{H}_i^{(2)} \\
 &\overset{(c)}{\Rightarrow}  (\tilde{H}_i^{(2)})^{+}  \tilde{H}_i^{(2)}  (\tilde{H}_i^{(2)})^{+} \geq   (\tilde{H}_i^{(2)})^{+}  \tilde{H}_i^{(2)}  (\tilde{H}_i^{(1)})^{+} \tilde{H}_i^{(2)}  (\tilde{H}_i^{(2)})^{+} \\
 &\overset{(d)}{=} \text{Condition $(ii)$ in \eqref{eq:priv_order}},
 \end{align*}
 where $(b)$ follows from \cite[Lemma 1]{REH:78}, and $(c)$, $(d)$ follow from facts that $(\tilde{H}_i^{(k)})^{+}$ is symmetric and $(\tilde{H}_i^{(k)})^{+} \tilde{H}_i^{(k)}= \tilde{H}_i^{(k)} (\tilde{H}_i^{(k)})^{+}$. Thus, both conditions in \eqref{eq:priv_order} are satisfied and $\mc{M}_i^{(2)} \geq \mc{M}_i^{(1)}$.
\end{proof}

We conclude this section by showing that the privacy mechanism in
\eqref{eq:limited_meas} exhibits an intuitive post-processing
property. It implies that if we further limit the measurements
produced by a privacy mechanism, then this operation cannot decrease
the privacy of the measurements. This post-processing property also
holds in the differential privacy framework
\cite{JC-GED-SH-JL-SM-GJP:16}.
\begin{lemma} (\textbf{Post-processing increases privacy}) Consider
  two privacy mechanisms $\mc{M}_i^{(1)}$ and $\mc{M}_i^{(2)}$, where
  $\mc{M}_i^{(2)}$ further limits the measurements of $\mc{M}_i^{(1)}$
  as:
\begin{align*}
\mc{M}_i^{(1)}: \qquad \tilde{y}_{i}^{(1)}(k) &= S_i^{(1)} y_i(k) + \tilde{r}_i^{(1)}(k)\\
\mc{M}_i^{(2)}: \qquad \tilde{y}_{i}^{(2)}(k) &= S \tilde{y}_i^{(1)}(k) + n_i(k),
\end{align*}
where $S$ is full row rank and $n_i(k) \sim \mc{N}(0,\Sigma_{n_i})$. Then, $\mc{M}_i^{(2)}$ is more private than $\mc{M}_i^{(1)}$.
\end{lemma}
\begin{proof}
It is easy to observe that $S_i^{(2)} = S S_i^{(1)}$ and $\tilde{r}_i^{(2)}(k) = S \tilde{r}_i^{(1)}(k) + n_i(k)$. Thus, 
\begin{align*}
\Sigma_{\tilde{r}_{i}}^{(2)} = S \Sigma_{\tilde{r}_{i}}^{(1)} S^{\trans} + \Sigma_{n_i} \geq S \Sigma_{\tilde{r}_{i}}^{(1)} S^{\trans},
\end{align*}
and the result follows from Lemma \ref{eq:suff_cond_priv_order}. 
\end{proof}

\blue{\begin{remark} (\textbf{Comparison with Differential Privacy
      (DP)}) Additive noise based privacy mechanisms have also been
    proposed in the framework of DP. Specifically, the notion of
    $(\epsilon,\delta)$-DP uses a zero-mean Gaussian noise
    \cite{JC-GED-SH-JL-SM-GJP:16}. Although the frameworks of DP and
    this paper use additive Gaussian noises, there are conceptual
    differences between the two. The DP framework distinguishes
    between the cases when a single subsystem is present or absent in
    the system, and tries to make the output statistically similar in
    both the cases. It allows access to arbitrary side information and
    does not involve any specific estimation algorithm. In contrast,
    our privacy framework assumes no side information and privacy
    guarantees are specific to the considered estimation
    procedure. Moreover, besides adding noise, our framework also
    allows for an additional means to vary privacy by sending fewer
    measurements, which is not feasible in the DP framework. \oprocend
  \end{remark}} 

\section{Detection performance vs privacy trade-off}
In this section, we present a trade-off between the attack detection
performance and privacy of the subsystems. As before, we focus on
detection for Subsystem $1$ and consider two measurement sharing
privacy mechanisms $\mc{M}_j^{(1)}$ and $\mc{M}_j^{(2)}$ for all other
subsystems $j\in\mc{S}_{-1}$. The trade-off is between the
\emph{detection performance of Subsystem 1} and the \emph{privacy
  level of all other subsystems}. We begin by characterizing the relation
between the detection parameters corresponding to these two sets of
privacy mechanisms.

\begin{theorem} \textbf{(Relation among the detection parameters of
    privacy mechanisms)} \label{thm:limit_info}  Let $\mc{M}_j^{(2)}$ be more private than $\mc{M}_j^{(1)}$ for all $j\in\mc{S}_{-1}$. Given any attack vector $a$, let $q^{(k)}$ and
  $\lambda^{(k)} = a^{\trans} {\Lambda}^{(k)}a$ denote the detection
  parameters under the privacy mechanisms
  $\left\{ \mc{M}_j^{(k)}\right\}_{j\in\mc{S}_{-1}}$, for
  $k=1,2$. Then, we have
\begin{flalign} \label{eq:det_param_relation}
\begin{split}
&(i) \:q^{(1)} \geq q^{(2)}  \quad \text{and,} \hspace{130pt}\\
&(ii) \:\lambda^{(2)} \mu_{max} \geq \lambda^{(1)} \geq \lambda^{(2)}\mu_{min} \geq \lambda^{(2)}, 
\end{split}
\end{flalign}

\noindent where $\mu_{max}$ and $\mu_{min}$ are the largest and
smallest generalized eigenvalues of
$({\Lambda}^{(1)},{\Lambda}^{(2)})$, respectively.
\end{theorem}
\begin{proof}
From \eqref{eq:limited_meas}, \eqref{eq:share_meas_1} and \eqref{eq:share_meas_2}, for $k=1,2$, we have
\begin{align*}
H^{(k)} &= I_T \otimes \text{diag}\left(S_{2}^{(k)}C_2,\cdots,S_{N}^{(k)}C_N\right) = S_{-1}^{(k)}H, \\
\Sigma_{v_{R}}^{(k)} &= S_{-1}^{(k)} \Sigma_{v_{R}} (S_{-1}^{(k)})^{\trans} +  \Sigma_{\tilde{r}_{-1}}^{(k)} >0\quad \text{where,}\\
S_{-1}^{(k)} &= I_T \otimes \text{diag}\left(S_{2}^{(k)},\cdots,S_{N}^{(k)}\right), \\
\Sigma_{\tilde{r}_{-1}}^{(k)}  &= I_T \otimes \text{diag}\left(\Sigma_{\tilde{r}_{2}}^{(k)}, \cdots,\Sigma_{\tilde{r}_{N}}^{(k)}\right) \geq 0.
\end{align*}
Since $\mc{M}_j^{(2)}\geq \mc{M}_j^{(1)}$ for all $j\in\mc{S}_{-1}$, the first condition in \eqref{eq:priv_order} results in
\begin{align*}
\text{Im}\!\left(\!(S_{-1}^{(1)})^{\trans}\right) \!\!\supseteq\! \text{Im}\!\left(\!(S_{-1}^{(2)})^{\trans}\right)\! \Rightarrow \! \text{Im}\!\left(\! (H^{(1)})^{\trans}\right)\! \!\supseteq\! \text{Im}\!\left(\! (H^{(2)})^{\trans}\right).
\end{align*}
From \eqref{eq:x_estimate}, we have $\tilde{H}^{(k)} = (H^{(k)})^{\trans} (\Sigma_{v_{R}}^{(k)})^{-1} H^{(k)}$. Since $\text{Null}(\tilde{H}^{(k)}) = \text{Null}(H^{(k)})$, from \eqref{eq:left_null_matrix}, it follows that\\ $\text{Im}(M^{(1)}) \supseteq \text{Im}(M^{(2)})$. Recalling from \eqref{eq:dist_H1} that $q^{(k)} = \text{Rank}((M^{(k)})^{\trans} F_a)$, it follows that $q^{(1)}\geq q^{(2)}$.

Since $\text{Im}(M^{(1)}) \supseteq \text{Im}(M^{(2)})$, we have
$M^{(2)} = M^{(1)}P$ for some full column rank matrix $P$. Let
$Z\triangleq F_x^{\trans} M^{(1)}P$. From \eqref{eq:z_noise_var}, we
have
\begin{align}
\Sigma_{v_P}^{(2)} &= (M^{(2)})^{\trans} \Sigma_{v_L} M^{(2)} + (M^{(2)})^{\trans}F_x (\tilde{H}^{(2)})^{+}F_x^{\trans} M^{(2)}, \nonumber\\ \label{eq:det_par_comp_pf_0}
 &= P^{\trans} \Sigma_{v_P}^{(1)} P + \underbrace{Z^{\trans} [(\tilde{H}^{(2)})^{+}- (\tilde{H}^{(1)})^{+}]Z}_{\normalfont \triangleq E}.
\end{align}
Next, we show that $E\geq0$. . Using $M^{(2)} = M^{(1)}P$, and using
\eqref{eq:left_null_matrix} for both $\{M^{(k)}, \tilde{H}^{(k)}\}$,
$k=1,2$, we have
\begin{align} \label{eq:det_par_comp_pf_1} Z^{\trans}
  (\tilde{H}^{(1)})^{+}\tilde{H}^{(1)} =
  Z^{\trans}(\tilde{H}^{(2)})^{+}\tilde{H}^{(2)}.
\end{align}
Thus, we get
\begin{align} \label{eq:error_temp_pf}
&E = Z^{\trans} [(\tilde{H}^{(2)})^{+}-(\tilde{H}^{(1)})^{+} \tilde{H}^{(1)}(\tilde{H}^{(1)})^{+} \tilde{H}^{(1)}(\tilde{H}^{(1)})^{+}] Z\nonumber \\
& = \! Z^{\trans} [(\tilde{H}^{(2)})^{+}\!\!-\!(\tilde{H}^{(2)})^{+} \!\tilde{H}^{(2)}(\tilde{H}^{(1)})^{+}\!(\tilde{H}^{(2)})^{+} \! \tilde{H}^{(2)}] Z
\end{align}
where the last inequality follows from \eqref{eq:det_par_comp_pf_1}
and the fact that
$\tilde{H}^{(k)}(\tilde{H}^{(k)})^{+} =(\tilde{H}^{(k)})^{+}
\tilde{H}^{(k)}$. Next, we have,
\begin{subequations}
\begin{align}
&\tilde{H}^{(k)}\! =\! I_T\! \otimes \!\text{diag}\Big[\!(S_2^{(k)}C_2)^{\trans}\! (S_2^{(k)} \Sigma_{v_2} (S_2^{(k)})^{\trans}\! \!+ \!\Sigma_{\tilde{r}_2}^{(k)})^{-1} \!S_2^{(k)}C_2, \nonumber\\
&\cdots, (S_N^{(k)}C_N)^{\trans} \left(S_N^{(k)} \Sigma_{v_N} (S_N^{(k)})^{\trans} + \Sigma_{\tilde{r}_N}^{(k)}\right)^{-1} S_N^{(k)}C_N\Big]\nonumber\\
 &\!=\! \Pi^{\trans} \text{diag}\Big[\! I_T\! \otimes \!(S_2^{(k)}C_2)^{\trans}(S_2^{(k)} \Sigma_{v_2} (S_2^{(k)})^{\trans} \!\!+\! \Sigma_{\tilde{r}_2}^{(k)})^{-1} S_2^{(k)}C_2,\nonumber\\
&\cdots,  I_T \!\otimes \!(S_N^{(k)}C_N)^{\trans}\! \left(\!S_N^{(k)} \Sigma_{v_N} (S_N^{(k)})^{\trans} \!+\! \Sigma_{\tilde{r}_N}^{(k)}\right)^{-1}\! S_N^{(k)}C_N\Big]\Pi\nonumber\\
&= \Pi^{\trans} \text{diag}\left[\tilde{H}_2^{(k)}, \cdots, \tilde{H}_N^{(k)}\right] \Pi \quad \text{and,} \label{eq:permuted_H_til}\\ \label{eq:permuted_H_til_pinv}
&(\tilde{H}^{(k)})^{+}  = \Pi^{\trans} \text{diag}\left[(\tilde{H}_2^{(k)})^{+}, \cdots, (\tilde{H}_N^{(k)})^{+}\right] \Pi,
\end{align}
\end{subequations}
where $\Pi$ is a permutation matrix with $\Pi^{-1} =
\Pi^{\trans}$. Substituting \eqref{eq:permuted_H_til} and
\eqref{eq:permuted_H_til_pinv} in \eqref{eq:error_temp_pf}, we have
\begin{align*}
E &= Z^{\trans}\Pi^{\trans} \text{diag}\Big[ (\tilde{H}_2^{(2)})^{+} - \mc{P}_2^{(2)}(\tilde{H}_2^{(1)})^{+} \mc{P}_2^{(2)}, \cdots \\
&(\tilde{H}_N^{(2)})^{+} - \mc{P}_N^{(2)}(\tilde{H}_N^{(1)})^{+} \mc{P}_N^{(2)} \Big] \Pi Z \overset{(a)}{\geq} 0,
\end{align*}
where $(a)$ follows from the second condition in \eqref{eq:priv_order}
for all $j\in\mc{S}_{-1}$. Next, from \eqref{eq:dist_H1}, we have,
\begin{align*}
\Lambda^{(2)} &= F_a^{\trans}M^{(2)} (\Sigma_{v_P}^{(2)})^{-1} (M^{(2)})^{\trans} F_a \\
 & \overset{\eqref{eq:det_par_comp_pf_0}}{=} F_a^{\trans}M^{(1)} \underbrace{P (P^{\trans} \Sigma_{v_P}^{(1)} P + E)^{-1} P^{\trans}}_{\triangleq Y} (M^{(1)})^{\trans} F_a \\
 & \overset{(b)}{\leq} F_a^{\trans}M^{(1)}  (\Sigma_{v_P}^{(1)})^{-1} (M^{(1)})^{\trans} F_a  = \Lambda^{(1)},\\
 \Rightarrow \lambda^{(1)} &= a^{\trans} \Lambda^{(1)} a  \geq   a^{\trans} \Lambda^{(2)} a =  \lambda^{(2)},
\end{align*}
where $(b)$ follows from Lemma \ref{lem:sigma_ordering} \blue{in the Appendix}, and the facts
that $E\geq 0$ and $P$ is full column rank. Finally, the second
condition in \eqref{eq:det_param_relation} follows from Lemma
\ref{lem:qcqp} \blue{in the Appendix}, and the proof is complete.
\end{proof}

Theorem \ref{thm:limit_info} shows that when the subsystems
$j\in\mc{S}_{-1}$ share measurements with Subsystem $1$ using more
private mechanisms, both the number of processed measurements and the
SNR reduce. This has implications on the detection performance of
Subsystem $1$, as explained next. To compare the performance
corresponding to the two sets of privacy mechanisms, we select the
same false alarm probability $P_F$ for both the cases and compare the
detection probability. Theorem \ref{thm:limit_info} and Lemma
\ref{lem:Pd_dependence} imply that $P_D(q^{(2)},\lambda^{(2)},P_F)$
can be greater or smaller than $P_D(q^{(1)},\lambda^{(1)},P_F)$
depending on the actual values of the detection parameters. In fact,
ignoring the dependency on $P_F$ since it is same for both cases, we
have
\begin{align*}
&P_D(q^{(2)},\lambda^{(2)}) - P_D(q^{(1)},\lambda^{(1)}) =\\ 
&\! \underbrace{P_D(q^{(2)}\!,\lambda^{(2)}\!)\! -\! P_D(q^{(2)}\!,\lambda^{(1)}\!)}_{\leq 0}
+ \underbrace{P_D(q^{(2)}\!,\lambda^{(1)}\!)\! -\! P_D(q^{(1)}\!,\lambda^{(1)}\!).}_{\geq 0}
\end{align*}
Intuitively, if the decrease in $P_D$ due to the decrease in the
SNR\footnote{Note that the SNR depends upon the attack vector $a$ (via
  \eqref{eq:dist_H1}), which we do not know a-priori. Thus, depending
  on the actual attack value, the SNR can take any positive value.}
($\lambda^{(1)} \rightarrow \lambda^{(2)}$) is larger than the
increase in $P_D$ due to the decrease in the number of measurements
($q^{(1)}\rightarrow q^{(2)}$), then the the detection performance
decreases, and vice-versa. \blue{The next result formalizes this intuition.}

\blue{\begin{theorem} \textbf{(Condition for
      trade-off)}\label{thm:tradeoff_cond} Consider the setup in
    Theorem \ref{thm:limit_info}, and let the detection probability be
    given in \eqref{eq:Pd}. Then, for a given $P_F$, a
    security-privacy trade-off exists if and only if
  \begin{align*}
  P_D(q^{(2)},\lambda^{(2)},P_F) \leq P_D(q^{(1)},\lambda^{(1)},P_F).
  \end{align*}
 \end{theorem}}
 \blue{The above result presents an analytical condition for the trade-off. When this condition is violated, a counter trade-off exists.} This is an interesting and counter-intuitive trade-off between the
detection performance and privacy/ information sharing, and it implies
that, in certain cases, sharing less information can lead to a better
detection performance. This phenomenon occurs because the GLRT for the
considered hypothesis testing problem is a sub-optimal test, as
discussed before.

\blue{\subsection{Privacy using only noise}
In this subsection, we analyze the special case when the subspace of the shared
measurements is fixed, and the privacy level can be varied by changing only the
noise level.} We begin by comparing the detection performance corresponding to two
privacy mechanisms that share the same subspace of measurements.

\begin{corollary} \textbf{(Strict security-privacy trade-off
    )} \label{cor:limit_meas_via_add_noise} Consider two privacy
  mechanisms $\mc{M}_j^{(2)} \geq \mc{M}_j^{(1)}$ such that
  $\textnormal{Im}\left((S_j^{(2)})^{\trans}\right) =
  \textnormal{Im}\left((S_j^{(1)})^{\trans}\right)$ for
  $j\in\mc{S}_{-1}$. Let ($q^{(k)},\lambda^{(k)}$) denote the
  detection parameters of Subsystem 1 under the privacy mechanisms
  $\left\{ \mc{M}_j^{(k)}\right\}_{j\in\mc{S}_{-1}}$, for
  $k=1,2$. Then, for any given $P_F$, we have
\begin{align*} 
P_D(q^{(2)},\lambda^{(2)},P_F) \leq P_D(q^{(1)},\lambda^{(1)},P_F).
\end{align*}
\end{corollary}
\begin{proof}
  Since the mechanisms share the same subspace of measurements, from
  the proof of Theorem \ref{thm:limit_info}, we have
\begin{align*}
\text{Im}\!\left(\!(S_{-1}^{(1)})^{\trans}\right) \!\! &=\! \text{Im}\!\left(\!(S_{-1}^{(2)})^{\trans}\right)\! \Rightarrow \! \text{Im}\!\left(\! (H^{(1)})^{\trans}\right)\! \! = \! \text{Im}\!\left(\! (H^{(2)})^{\trans}\right) \\
& \Rightarrow \! \text{Im}\!\left(\! M^{(1)}\right)\! \! = \! \text{Im}\!\left(\! M^{(2)}\right) \Rightarrow q^{(1)} = q^{(2)}.
\end{align*}
The fact that $\lambda^{(1)}\geq \lambda^{(2)}$ follows from Theorem \ref{thm:limit_info}, and the result then follows from Lemma \ref{lem:Pd_dependence}.
\end{proof}

The above result implies that there is strict trade-off between
privacy and detection performance when the subspace of the shared
measurements is fixed and the privacy level is varied by changing the
noise level. In this case, more private mechanisms result in a poorer
detection performance, and vice-versa.

\blue{ Corollary \ref{cor:limit_meas_via_add_noise} qualitatively
  captures the security-privacy trade-off. Next, we present a
  quantitative analysis that determines the best possible detection
  performance subject to a given privacy level. Note that since the
  subspace of the shared measurements is fixed, the detection
  parameter $q$ is also fixed, as well as $P_F$ and the attack
  $a$. Thus, according to Lemma \ref{lem:Pd_dependence}, the detection
  performance can be improved by increasing
  $\lambda = a^{\trans} \Lambda a$. Intuitively, $\lambda$ is large
  (irrespective of $a$) if $\Lambda$ is large, or when $\Lambda^{+}$
  is small (in a positive-semidefinite sense).\footnote{Minimization
    of $\Lambda^{+}$ allows us to formulate a semidefinite
    optimization problem, as we show later.} Further, the privacy
  level is quantified by the error covariance in
  \eqref{eq:err_cov}. Based on this, we formulate the following
  optimization problem:\footnote{This problem corresponds to Subsystem
    $1$. A similar problem can be formulated for the whole system
    whose cost is the sum of the costs of the individual subsystems.}

\begin{align} \label{eq:opt prob}
&\underset{\textstyle \Sigma_{\tilde{r}_2}\geq0,\cdots, \Sigma_{\tilde{r}_N}\geq0} {\min} && \text{Tr}(\Lambda^{+}) \\\nonumber
&\qquad\qquad \text{s.t.}   &&\text{Tr}(\Sigma_{e_i}) \geq \epsilon_i',       \quad  i=2,\cdots,N,
\end{align} 
where $\epsilon_i'>0$ is the minimum desired privacy level of Subsystem $i$. The design variables of the above optimization problem are the positive semi-definite noise covariance matrices $\Sigma_{\tilde{r}_2},\cdots, \Sigma_{\tilde{r}_N}$. Next, we show that under some mild assumptions, \eqref{eq:opt prob} is a semidefinite optimization problem.

\begin{lemma}
  Assume that $F(I)$ in \eqref{eq:F_matrix} and $C_i$ for
  $i=1,\cdots,N$ are full row rank. Let
  $D_1 = \sum_{i=1}^{T} D_{1,i}$, where the matrices
  $D_{1,i}\in\real^{n_1 \times n_1}$, $i=1,\cdots, T$ are the block
  diagonal elements of $(M^{\trans} F_a)^{+} M^{\trans} F_a$. Further,
  let
\begin{align*}
K_1 &= B_1 (S_{-1}C_{-1})^{+}, \qquad L_1 = K_1^{\trans} D_1 K_1,\\
l_1 &= \textnormal{Tr}\left[(M^{\trans} F_a)^{+}  M^{\trans }\Sigma_{v_L} M ((M^{\trans} F_a)^{+})^{\trans}\right] \\
& + \textnormal{Tr}\left[ (M^{\trans} F_a)^{+} M^{\trans} F_a \left[I_T \otimes \left(K_1 S_{-1} \Sigma_{v_{-1}}S_{-1}^{\trans} K_1^{\trans}  \right)  \right] \right], \\
 g_i &= \textnormal{Tr}\left[H_i^{+} [I_T \otimes (S_i\Sigma_{v_i}S_i^{\trans})](H_i^{+})^{\trans}\right], \\
  G_i &= ((S_iC_i)^{+})^{\trans} (S_iC_i)^{+}.
\end{align*}
Then, $\textnormal{Tr}(\Lambda^{+}) = l_1 + \textnormal{Tr}(L_1 \Sigma_{\tilde{r}_{-1}})$ and $\textnormal{Tr}(\Sigma_{e_i}) = g_i + T \: \textnormal{Tr}(G_i \Sigma_{\tilde{r}_i})$, where $\Sigma_{\tilde{r}_{-1}} = \textnormal{diag}(\Sigma_{\tilde{r}_2},\cdots,\Sigma_{\tilde{r}_N})$.
\end{lemma}
\begin{proof}
From \eqref{eq:err_cov}, $\Sigma_{e_i} = \tilde{H}_i^{+}$, where $\tilde{H}_i = H_i^{\trans} \Sigma_{r_i}^{-1} H_i$ and $H_i = I_T \otimes S_iC_i$. Since, $S_i$ and $C_i$ are assumed to be full row rank, $H_i$ is full row rank. Next, we have
\begin{align*}
\tilde{H}_i^{+} &\overset{(a)}{=} H_i^{+} \Sigma_{r_i}( H_i^{+})^{\trans} \\
& \overset{\eqref{eq:limited_meas_tag}}{=} \underbrace{H_i^{+} [I_T \otimes (S_i\Sigma_{v_i}S_i^{\trans})](H_i^{+})^{\trans}}_{U_i}  + H_i^{+}[I_T \otimes \Sigma_{\tilde{r}_i}]( H_i^{+})^{\trans} \\
& \overset{\eqref{eq:limited_meas_tag}}{=} U_i + I_T \otimes [(S_iC_i)^{+}  \Sigma_{\tilde{r}_i} ((S_iC_i)^{+})^{\trans} ]\\
& \Rightarrow \textnormal{Tr}(\tilde{H}_i^{+}) = g_i + T \:\textnormal{Tr}(G_i \Sigma_{\tilde{r}_i} ),
\end{align*}
where $(a)$ follows from Lemma \ref{lem:H_til_pinv} in the Appendix. 

Next, from \eqref{eq:dist_H1}, $\Lambda = M_1^{\trans} \Sigma_{v_P}^{-1} M_1$ where $M_1=  M^{\trans} F_a$. Since $M^{\trans}$ and $F_a = F(I)$ are assumed to be full row rank, $M_1$ is full row rank. Next, we have
\begin{align*}
\Lambda^{+} & \overset{(b)}{=} M_1^{+}  \Sigma_{v_P} (M_1^{+})^{\trans} \\
& \overset{\eqref{eq:z_noise_var},\eqref{eq:block_meas}, (b)}{=}  M_1^{+}  M^{\trans }\Sigma_{v_L} M (M_1^{+})^{\trans} \\
&+ M_1^{+} M_1 (I_T\otimes B_1) H^{+}\Sigma_{v_R} (H^{+})^{\trans} (I_T\otimes B_1^{\trans}) M_1^{\trans}(M_1^{+})^{\trans} \\
& \overset{\eqref{eq:share_meas_2}}{=} M_1^{+}  M^{\trans }\Sigma_{v_L} M (M_1^{+})^{\trans} \\
 &+ M_1^{+} M_1 \left[I_T \otimes \left(K_1 S_{-1} \Sigma_{v_{-1}}S_{-1}^{\trans}  K_1^{\trans}  \right) \right] M_1^{+} M_1 \\
 &+ M_1^{+} M_1 \left[I_T \otimes \left(K_1  \Sigma_{\tilde{r}_{-1}} K_1^{\trans}  \right) \right] M_1^{+} M_1,\\
 \Rightarrow &\textnormal{Tr}(\Lambda^{+}) \overset{(c)}{=} l_1 + \textnormal{Tr}(D_1 K_1 \Sigma_{\tilde{r}_{-1}} K_1^{\trans})= l_1 + \textnormal{Tr}(L_1 \Sigma_{\tilde{r}_{-1}}).
\end{align*}
where $(b)$ follows from Lemma \ref{lem:H_til_pinv} in the Appendix, and $(c)$ follows from the definition of $D_1$ and trivial algebraic manipulation.
\end{proof}

Using the above theorem, \eqref{eq:opt prob} is equivalent to the following semidefinite optimization problem \cite{LV-SB:06}
\begin{align} \label{eq:opt prob1}
&\underset{\textstyle \Sigma_{\tilde{r}_2}\geq0,\cdots, \Sigma_{\tilde{r}_N}\geq0} {\min} &&  \text{Tr}(L_1 \Sigma_{\tilde{r}_{-1}}) + l_1\\\nonumber
&\qquad\qquad \text{s.t.}   &&\text{Tr}(G_i \Sigma_{\tilde{r}_i}) \geq \frac{\epsilon_i'-g_i}{T}:=\epsilon_i\geq 0,     
\end{align}
for $i=2,3,\cdots,N$, which can be solved using standard semidefinite
optimization algorithms \cite{LV-SB:06}. This analysis allows us to
design optimal noisy privacy mechanisms.}

\section{Simulation Example}
\blue{We consider a power network model of the IEEE 39-bus test case
  \cite{Athay:79} consisting of $10$ generators interconnected by
  transmission lines whose resistances are assumed to be
  negligible. Each generator is modeled according to the following
  second-order swing equation \cite{PK:94}:
\begin{align} \label{eq:gen_dynamics}
M_i \ddot{\theta_i} + D_i \dot{\theta_i} = P_i- \sum_{k=1}^{n}\frac{E_i E_k}{X_{ik}}\sin(\theta_i-\theta_k),
\end{align}
where $\theta_i, M_i, D_i,E_i$ and $P_i$ denote the rotor angle, moment of inertia, damping coefficient, internal voltage and mechanical power input of the $i^{\text{th}}$ generator, respectively. Further, $X_{ij}$ denotes the reactance of the transmission line connecting generators $i$ and $j$ ($X_{ij}=\infty$ if they are not connected). We linearize \eqref{eq:gen_dynamics} around an equilibrium point to obtain the following collective small-signal model:
\begin{align} \label{eq:power_net_dyn}
\begin{bmatrix} d\dot{\theta} \\ d\ddot{\theta} \end{bmatrix}&= \underbrace{\begin{bmatrix} 0 & I\\-M^{-1}L & -M^{-1}D \end{bmatrix}}_{\tilde{A}_c} \underbrace{\begin{bmatrix} d\theta \\d \dot{\theta}\end{bmatrix}}_{\tilde{x}(t)} + \underbrace{\begin{bmatrix} 0 \\ M^{-1}B \end{bmatrix}}_{ \tilde{B}^a_c}\tilde{a}(t),
\end{align}
where $d\theta$ denotes a small deviation of $\theta=\begin{bmatrix} \theta_1 & \theta_2 & \cdots & \theta_{10}\end{bmatrix}^{T}$ from the equilibrium value, $M=\text{diag}(M_1,M_2,\cdots,M_{10})$, $D=\text{diag}(D_1,D_2,\cdots,D_{10})$, and $L$ is a symmetric Laplacian matrix given by 
\begin{align}
L_{ij} = \begin{cases} -\frac{E_i E_j}{X_{ij}} \cos(\theta_i-\theta_j) \quad &\text{for} \quad i\neq j,  \\ -\sum\limits_{\substack{j=1\\j\neq i}}^{n} L_{ij} \quad &\text{for} \quad i=j. \end{cases}
\end{align}

Further, $\tilde{a}$ models small malicious alterations (attacks) in the mechanical power input of the generators that need to be detected. We assume that generators $\{1,4,8\}$ are under attack. Thus, $B = \begin{bmatrix} {\bf e_1, e_4,e_8} \end{bmatrix}$, where ${\bf e_j}$ denotes the $j^{\text{th}}$ canonical vector. We assume that the power network is divided into 3 subsystems consisting of generators $\{1,2,3\}$, $\{4,5,6,7\}$ and $\{8,9,10\}$. Accordingly, we permute the state vector in \eqref{eq:power_net_dyn} using a permutation matrix $\Pi$ such that $\Pi \tilde{x} = x = [x_1^{\trans},x_2^{\trans},x_3^{\trans}]^{\trans}$, where $x_i$ consists of rotor angles and velocities of all generators in Subsystem $i$. The transformed system is given by $\dot{x} = A_c x + B_c^a \tilde{a}(t)$, where $A_c = \Pi \tilde{A}_c \Pi^{-1}$ and $B^a_c = \Pi \tilde{B}^a_c$. Next, we sample this continuous time system with sampling time $T_s = 0.1$ to obtain a discrete-time system $x(k+1) =Ax(k) + B^a\tilde{a}(k)$ with  $A = e^{A_cT_s}$ and $B^a=\left(\int\limits_{t=0}^{T_s} e^{A_c\tau}d\tau\right)B_c^a$. We assume that the discrete-time process dynamics are affected by process noise according to \eqref{eq:ss_state}. The rotor angle and the angular velocity of all generators are measured using Phasor Measurement Units (PMUs) according to the noisy model \eqref{eq:ss_output}. The time horizon is $T=3$.

The generator voltage and angle values are obtained from \cite{Athay:79}. We fix the damping coefficient for each generator as $10$, and the moment of inertia values are chosen as $M=[70,10,40,30,70,30,90,80,40,50]$. The reactance matrix $X$ is generated randomly, where each entry of $X$ is distributed independently according to $\mathcal{N}(0, 0.01)$. We focus on the attack detection for Subsystem $1$, where Subsystems $2$ and $3$ use privacy mechanisms to share their measurements with
Subsystem 1. The parameters of Subsystem 1 can be extracted from $A, B^a$ as $A = \begin{bmatrix} A_1 \:\; B_1 \\ *\end{bmatrix}$ and $B^a = \text{blockdiag}(B_1^a,*,*)$. The noise covariances are $\Sigma_{w_1}=0.5I_6$ and $\Sigma_{v_1}= \Sigma_{v_3} =I_4$ and $\Sigma_{v_2} =0.5I_6$. 

We consider the following three cases of privacy
mechanisms for Subsystems 2 and 3:
\begin{itemize}
\item $\mc{M}^{(0)} = \{\mc{M}_2^{(0)}, \mc{M}_3^{(0)}\}$: Subsystems
  2 and 3 do not use any privacy mechanisms and share actual
  measurements, i.e.,
  $S_2 = I_8, S_3 = I_6, \Sigma_{\tilde{r}_{2}} = 0$, and
  $\Sigma_{\tilde{r}_{3}} = 0$.
\item $\mc{M}^{(1)}$: Subsystem 2 does not use any privacy mechanism
  ($S_2 = I_8, \Sigma_{\tilde{r}_{2}} = 0 $) while Subsystem 3 shares
  noisy measurements of generators $\{8,9\}$\\
  ($S_3 = \left[ {\bf e_1,e_2,e_3,e_4}\right]^{\trans}$,
  $\Sigma_{\tilde{r}_{3}} = I_4$).
\item $\mc{M}^{(2)}$: Subsystems 2 and 3 share noisy measurements of
  generators $\{4,5,6\}$ and $\{8,9\}$, respectively.
  ($S_2 = \left[{\bf e_1,e_2,e_3,e_4, e_5,e_6}\right]^{\trans}, S_3 =
  \left[{\bf e_1,e_2,e_3,e_4}\right]^{\trans}, \Sigma_{\tilde{r}_{2}}
  = I_6$, and $\Sigma_{\tilde{r}_{3}} = I_4$).
\end{itemize}
}
Using Lemma \ref{eq:suff_cond_priv_order}, it can be easily verified
that the following privacy ordering holds:
$\mc{M}^{(2)}>\mc{M}^{(1)} >\mc{M}^{(0)}$. Recall that the detection
performance is completely characterized by $P_F$ and the detection
parameters $(q,\lambda)$. We choose $P_F = 0.05$ for all the
cases. Let $(q^{(k)},\lambda^{(k)})$, $k = 0,1,2$ denote the detection
parameters for the above three cases. Recall that the parameter $q$
depends only the system parameters, whereas the parameter $\lambda$
depends on the system parameters as well as the attack values. For the
above cases, we have $q^{(0)} = \blue{18}, q^{(1)} =\blue{12}$ and
$q^{(3)}=\blue{6}$. Recalling \eqref{eq:dist_H1}, the value of
$\lambda^{(k)} = a^{\trans} \Lambda^{(k)} a$ can lie anywhere between
$[0,\infty)$ depending on the attack value $a$. Thus, for simplicity,
we present the results in this section in terms of $\lambda^{(k)}$.

\begin{figure}[t]
  \centering
  \subfigure[]{
  \includegraphics[width=.46\columnwidth]{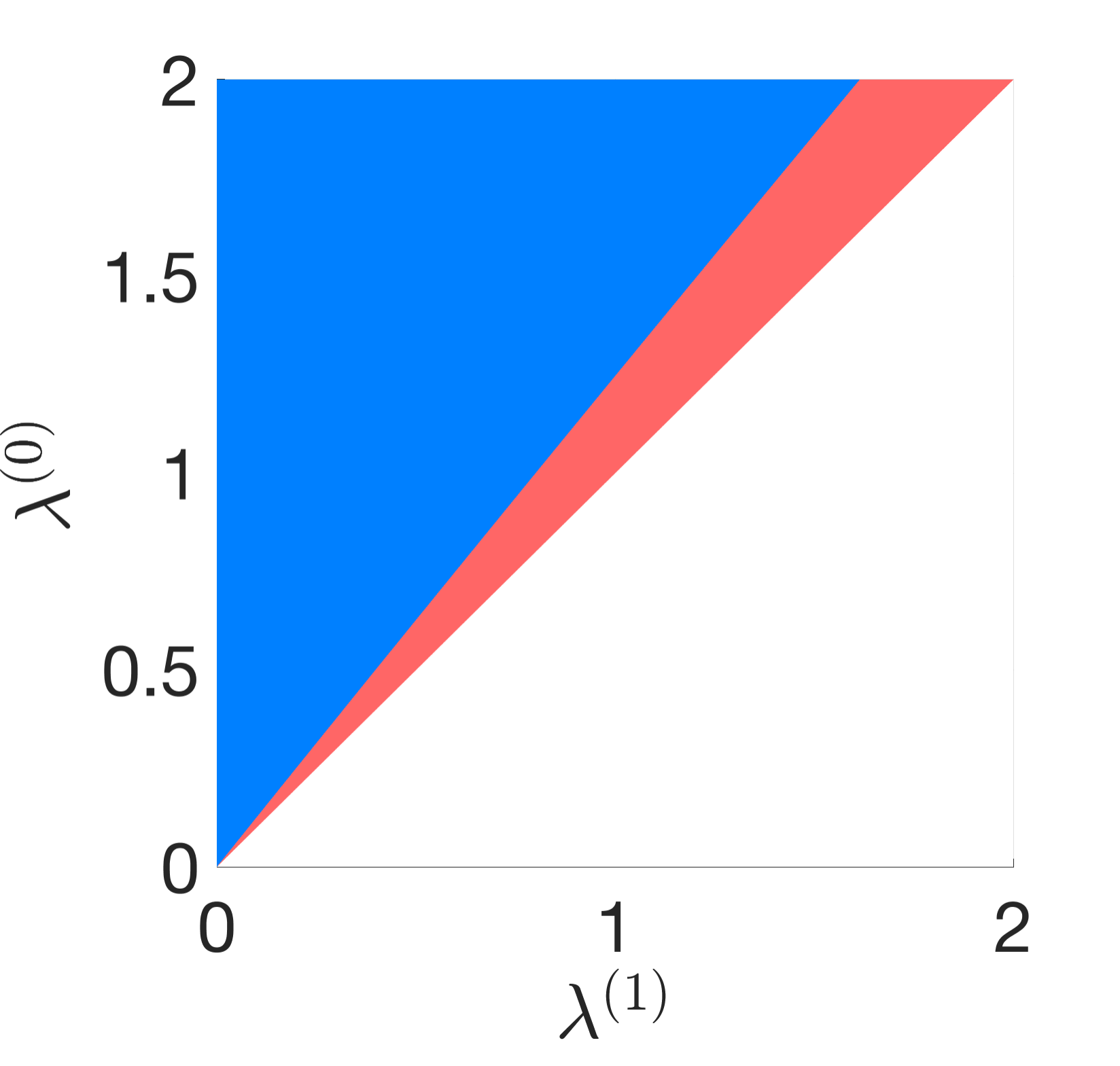} \label{fig:limit_meas_comp1}} 
  \subfigure[]{
  \includegraphics[width=.46\columnwidth]{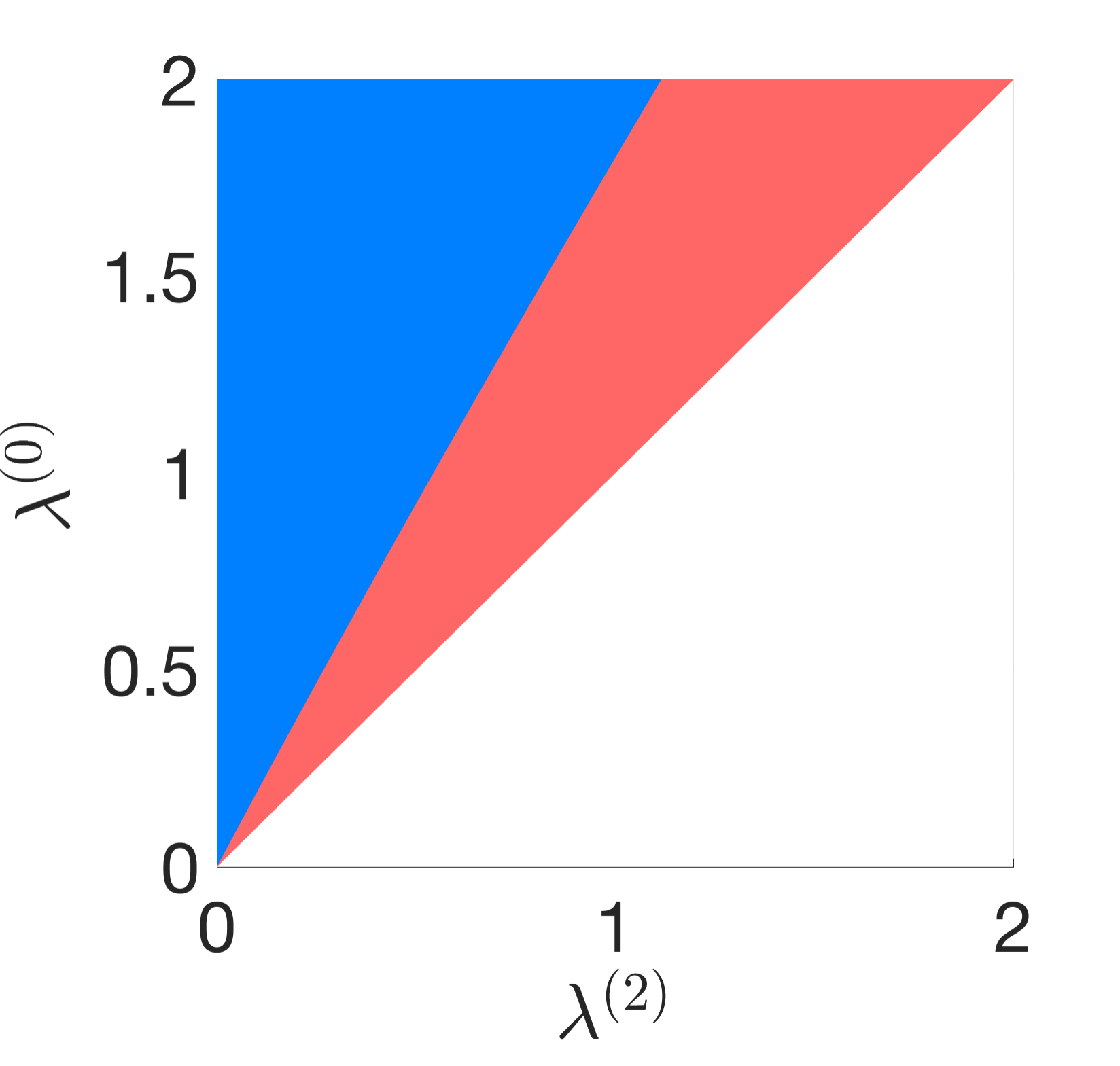} \label{fig:limit_meas_comp2}}
\caption{Comparison between detection performance of case 0 with: (a)
  case 1, and (b) case 2. In the blue region, case 0 performs better
  than case0/case 1, and vice versa in red region. Since
  $\lambda^{(0)} \geq \lambda^{(k)}$ for $k={1,2}$ (c.f. Lemma
  \ref{thm:limit_info}), the white region is inadmissible.}
  \label{fig:limit_meas_comp}
\end{figure}

We aim to compare the detection performance of case 0 with cases 1 and
2, respectively. We are interested in identifying the ranges of the
detection parameters for which one case performs better than the
other. As mentioned previously, the parameters $q^{(k)}$ are fixed for
the three cases, so we compare the performance for different values of
the parameter $\lambda^{(k)}$. Fig. \ref{fig:limit_meas_comp} presents
the performance comparison of case 0 with case 1
(Fig. \ref{fig:limit_meas_comp1}) and case 2
(Fig. \ref{fig:limit_meas_comp1}). Any point $(x,y)$ in the colored
regions are achievable by an attack, i.e., there exists an attack $a$
such that $a^{\trans} \Lambda^{(k)} a = x$ and
$a^{\trans} \Lambda^{(0)} a = y$, whereas the white region is
inadmissible (see \eqref{eq:det_param_relation}). The blue region
corresponds to the pairs $(\lambda^{(k)},\lambda^{(0)})$ for which
case 0 performs better than case $k$, i.e.,
$P_D(q^{(0)},\lambda^{(0)},P_F)\geq P_D(q^{(k)},\lambda^{(k)},P_F)$
for $k=1,2$. In the red region, case $k$ performs better that case 0,
$k=1,2$.

We observe that case 0 performs better than case $k$ if
$\frac{\lambda^{(0)}}{\lambda^{(k)}}$ is large, and vice versa. This
shows that if the attack vector $a$ is such that
$\frac{\lambda^{(0)}}{\lambda^{(k)}}$ is small, then the detection
performance corresponding to a more private mechanism
($\mc{M}^{(k)} >\mc{M}^{(0)}$) is better. This implies that there is
non-strict trade-off between privacy and detection performance.  This
counter-intuitive result is due to the suboptimality of the GLRT used
to perform detection, as explained before (c.f. discussion above
Remark \ref{rem:comp_vs_sim_test}). Further, we observe that the red
region of Fig. \ref{fig:limit_meas_comp2} is larger than (and
contains) the red region of Fig. \ref{fig:limit_meas_comp2}. This is
because $\mc{M}^{(2)}$ is more private than $\mc{M}^{(1)}$.

Next, we consider the case where Subsystems $2$ and $3$ implement
their privacy mechanisms by only adding artificial noise in
\eqref{eq:limited_meas}. Thus, $S_2 = I_8, S_3 = I_6$, and the
artificial noise covariances are given by
$\Sigma_{\tilde{r}_{2}} =\sigma^2 I_8$ and
$\Sigma_{\tilde{r}_{3}} = \sigma^2 I_6$. The attack on Subsystem $1$ (that is, on generator $1$) is
$\tilde{a}(k) = 2500$ for $k=0,1,2$. Clearly, as the noise
level $\sigma$ increases, the privacy level also
increases. Fig. \ref{fig:Pd_vs_sigma} shows the detection performance
of Subsystem $1$ for varying noise level $\sigma$. We observe that the
detection performance is a decreasing function of the noise level
(c.f. Corollary \ref{cor:limit_meas_via_add_noise}), implying a strict
trade-off between detection performance and privacy in this case.
\blue{Finally, we illustrate this strict trade-off by also explicitly solving the optimization problem \eqref{eq:opt prob1} and computing the optimal noise covariance matrices. We fix the same desired privacy level for Subsystems 2 and 3: $\epsilon_1 = \epsilon_2 = \epsilon$. Fig. \ref{fig:Opt_prob} shows that the optimal cost in \eqref{eq:opt prob1} increases with $\epsilon$, indicating that the detection performance decreases as privacy level increases.}

\begin{figure}[h!]
\centering
\includegraphics[width=\columnwidth]{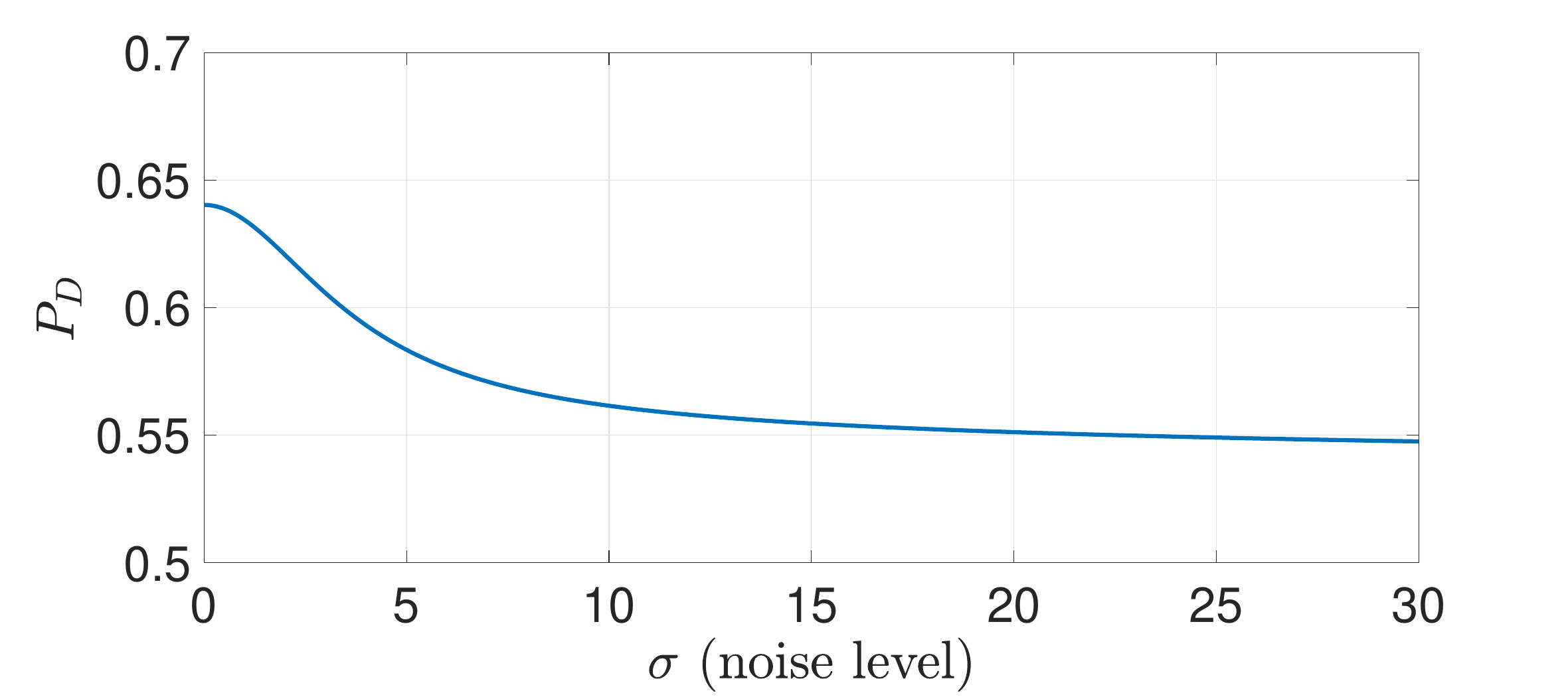}
\caption{Detection performance for varying level of  noise parameter $\sigma$.}
\label{fig:Pd_vs_sigma}
\end{figure}

\begin{figure}[h!]
\centering
\includegraphics[width=\columnwidth]{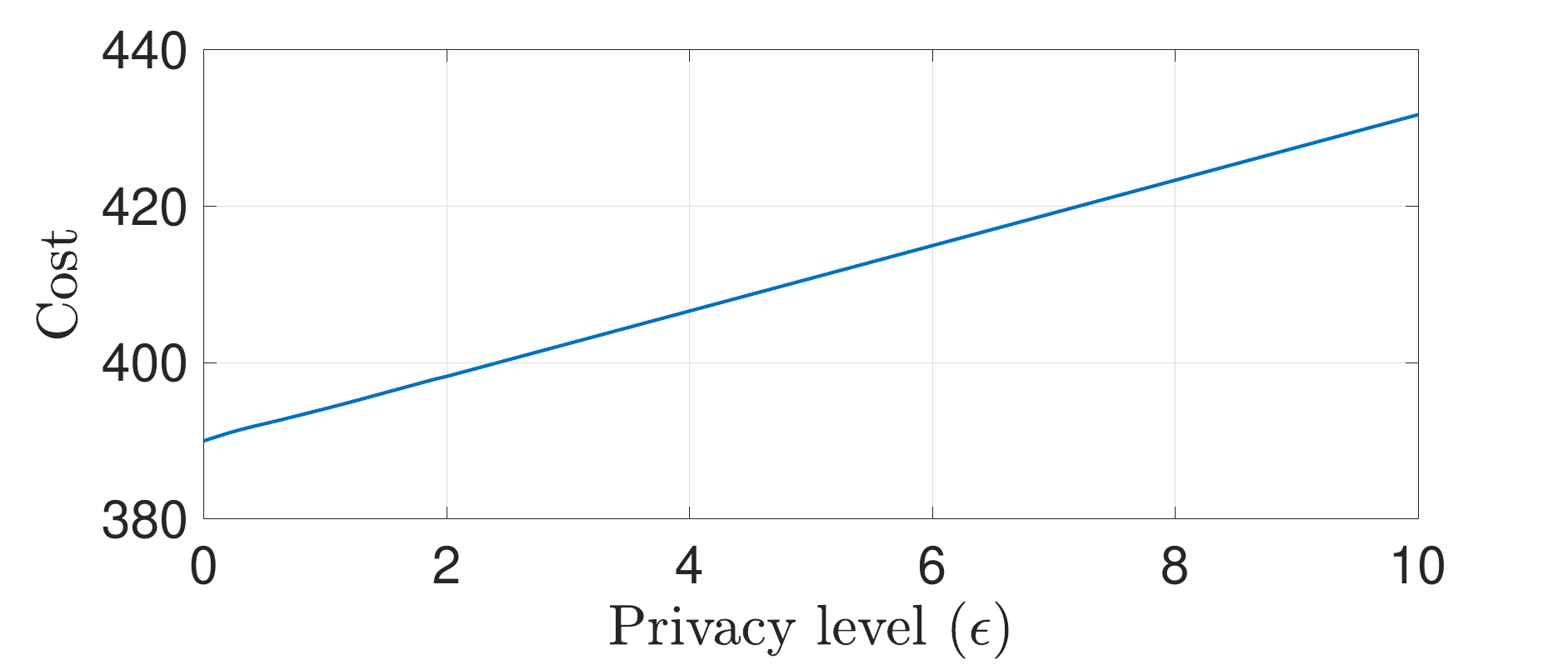}
\caption{Optimal cost of \eqref{eq:opt prob1} as a function of the privacy level $\epsilon$.}
\label{fig:Opt_prob}
\end{figure}

\section{Conclusion}
We study an attack detection problem in interconnected dynamical
systems where each subsystem is tasked with detection of local attacks
without any knowledge of the dynamics of other subsystems and their
interconnection signals. The subsystems share measurements among
themselves to aid attack detection, but they also limit the
amount and quality of the shared measurements due to privacy concerns. We
show that there exists a non-strict trade-off between privacy and
detection performance, and in some cases, sharing less measurements
can improve the detection performance. We reason that this
counter-intuitive result is due the suboptimality of the considered
$\chi^2$ test.

Future work includes exploring if this counter-intuitive trade-off
exist for alternative detection schemes (for instance, unknown-input
observers) and for other types of statistical tests. Also, recursive
schemes to compute the state estimates, eliminate interconnections and
compute the detection probability will be explored. \blue{Finally, privacy ordering of two mechanisms irrespective of their subspaces of shared measurement will be defined using suitable weighing matrix for each subspace.}

\setcounter{section}{1} 
\renewcommand{\thetheorem}{A.\arabic{theorem}}
\section*{APPENDIX}

\begin{lemma} \label{lem:WLS} 
The optimal solutions of the following weighted least squares problem:
\begin{align}\label{eq:WLS}
\underset{x}{\min} \quad J(x) = (y-Hx)^{\trans} \Sigma^{-1} (y-Hx),
\end{align}
with $\Sigma>0$ are given by
\begin{align} \label{eq:WLS_soln}
x^{*} = \tilde{H}^{+} H^{\trans} \Sigma^{-1} y + (I-\tilde{H}^{+} \tilde{H}) d,
\end{align}
where $\tilde{H} = H^{\trans} \Sigma^{-1} H,$ and $d$ is any real vector of appropriate dimension. Further, the optimal value of the cost is
\begin{align}\label{eq:WLS_cost}
J(x^{*}) = y^{\trans}( \Sigma^{-1}- \Sigma^{-1} H\tilde{H}^{+} H^{\trans} \Sigma^{-1})y.
\end{align} 
\end{lemma}

\begin{lemma} \label{lem:schur_comp} 
Let $\left[\begin{smallmatrix}  A & B \\ B^{\trans} & D\end{smallmatrix}\right]$ be a positive definite matrix with $A>0$, $D\geq 0$. Further, let $M\geq0$. Then,
\begin{align*}
\left[\begin{smallmatrix}  A & B \\ B^{\trans} & D\end{smallmatrix}\right]^{-1} \geq \left[\begin{smallmatrix}  (A+M)^{-1} & 0 \\ 0 & 0\end{smallmatrix}\right],
\end{align*}
\end{lemma}
\begin{proof}
Using the Schur complement, we have
\begin{align*}
\left[\begin{smallmatrix}  A & B \\ B^{\trans} & D\end{smallmatrix}\right]^{-1} = \left[\begin{smallmatrix} I & -A^{-1}B \\ 0 & I \end{smallmatrix}\right] \left[\begin{smallmatrix} A^{-1} & 0 \\ 0 & (D-B^{\trans} A^{-1} B)^{-1} \end{smallmatrix}\right] \left[\begin{smallmatrix}  I & 0 \\ -B^{\trans} A^{-1} & I \end{smallmatrix}\right], 
\end{align*}
where the Schur complement $D-B^{\trans} A^{-1} B>0$. Further,
\begin{align*}
\left[\begin{smallmatrix}  (A+M)^{-1} & 0 \\ 0 & 0\end{smallmatrix}\right] = \left[\begin{smallmatrix} I & -A^{-1}B \\ 0 & I \end{smallmatrix}\right] \left[\begin{smallmatrix} (A+M)^{-1} & 0 \\ 0 & 0 \end{smallmatrix}\right] \left[\begin{smallmatrix}  I & 0 \\ -B^{\trans} A^{-1} & I \end{smallmatrix}\right]
\end{align*}
Since $A+M \geq A$, $A^{-1} \geq (A+M)^{-1}$. Thus,
\begin{align*}
\left[\begin{smallmatrix} A^{-1}-(A+M)^{-1}  & 0 \\ 0 & (D-B^{\trans} A^{-1} B)^{-1} \end{smallmatrix}\right] \geq 0,
\end{align*}
and the result follows.
\end{proof}

\begin{lemma} \label{lem:sigma_ordering}
Let $\Sigma>0 \in\real^{n\times n}$ and $\Sigma_a \geq 0 \in\real^{m\times m},$ with $m\leq n$, and let $S\in\real^{n\times m}$ be full (column) rank. Then,
\begin{align}
\Sigma^{-1} \geq S (S^{\trans} \Sigma S + \Sigma_{a})^{-1} S^{\trans}.
\end{align}
\end{lemma}
\begin{proof}
  Since $S$ is full column rank, $S^{\trans} \Sigma S>0$,
  $S^{\trans} \Sigma S + \Sigma_{a}$ is invertible and
  $S^{+}S = I_{m} = S^{\trans} (S^{\trans})^{+}$. Thus,
  $I_n = \left[\begin{smallmatrix} S^{\trans} (S^{\trans})^{+} & 0 \\
      0 & I_{n-m} \end{smallmatrix}\right]$. Let
  $N\in\real^{n\times (n-m)}$ denote a matrix whose columns are the
  basis of $\text{Null}(S^{\trans})$. Then,
  $\left[\begin{smallmatrix} S^{\trans} (S^{\trans})^{+} &
      0 \end{smallmatrix}\right] = S^{\trans}
  \left[\begin{smallmatrix} (S^{\trans})^{+} &
      N \end{smallmatrix}\right] \triangleq S^{\trans} R.$ Since,
  $\text{Im}((S^{\trans})^{+}) = \text{Im}(S) \perp
  \text{Null}(S^{\trans})$, $R$ is non-singular. Let
  $T\triangleq \left[\begin{smallmatrix} 0 &
      I_{n-m}\end{smallmatrix}\right]R^{-1}$. Then, we have
  $I_n = \left[\begin{smallmatrix} S^{\trans} \\
      T \end{smallmatrix}\right] R = R \left[\begin{smallmatrix}
      S^{\trans} \\ T \end{smallmatrix}\right]$. Thus,
\begin{align*}
\Sigma^{-1} &= I_n^{\trans} (I_n \Sigma I_n^{\trans})^{-1} I_n \\
&= \left[\begin{smallmatrix} S & T^{\trans} \end{smallmatrix}\right] R^{\trans} \left( R \left[\begin{smallmatrix} S^{\trans} \\ T \end{smallmatrix}\right] \Sigma  \left[\begin{smallmatrix} S & T^{\trans} \end{smallmatrix}\right] R^{\trans} \right)^{-1}  R \left[\begin{smallmatrix} S^{\trans} \\ T \end{smallmatrix}\right] \\
& = \left[\begin{smallmatrix} S & T^{\trans} \end{smallmatrix}\right] \left(  \left[\begin{smallmatrix} S^{\trans} \\ T \end{smallmatrix}\right] \Sigma  \left[\begin{smallmatrix} S & T^{\trans} \end{smallmatrix}\right]  \right)^{-1}   \left[\begin{smallmatrix} S^{\trans} \\ T \end{smallmatrix}\right] \\
& = \left[\begin{smallmatrix} S & T^{\trans} \end{smallmatrix}\right] \left[\begin{smallmatrix} S^{\trans} \Sigma S & S^{\trans} \Sigma T^{\trans} \\ T \Sigma S & T \Sigma T^{\trans}\end{smallmatrix}\right] ^{-1} \left[\begin{smallmatrix} S^{\trans} \\ T \end{smallmatrix}\right], \quad \text{and} \\
S (S^{\trans} \Sigma S & + \Sigma_{a})^{-1} S^{\trans} = \left[\begin{smallmatrix} S & T^{\trans} \end{smallmatrix}\right] \left[\begin{smallmatrix} (S^{\trans} \Sigma S  + \Sigma_{a})^{-1} & 0 \\ 0 & 0\end{smallmatrix}\right] ^{-1} \left[\begin{smallmatrix} S^{\trans} \\ T \end{smallmatrix}\right].
\end{align*}
The result follows from Lemma \ref{lem:schur_comp}.
\end{proof}

\begin{lemma} \label{lem:qcqp} Let $M_1\geq M_2\geq 0$,
  $\lambda \geq 0$ and let $J(x) = x^{\trans} M_1 x$. Then, the
  maximum and minimum values of $J(x)$ subject to
  $x^{\trans} M_2 x = \lambda$ are given by $\lambda\mu_{max}$ and
  $\lambda \mu_{min}$ respectively, where $\mu_{max}$ and $\mu_{min}$
  are the largest and smallest generalized eigenvalues of $(M_1,M_2)$,
  respectively.
\end{lemma}
\begin{proof}
Consider the following optimization problem
\begin{align*}
\underset{x}{\min/\max} \quad J(x) = x^{\trans} M_1 x, \quad \text{subject to} \quad x^{\trans} M_2 x = \lambda.
\end{align*} 
The Lagrangian of this problem is given by
$l = x^{\trans} M_1 x - \mu (x^{\trans} M_2 x - \lambda)$, where
$\mu \in\real$ is the Lagrange multiplier. By differentiating $l$, the
first order optimality condition is given by $(M_1 -\mu M_2)
x=0$. Thus, $\mu$ is a generalized eigenvalue of $(M_1,M_2)$. Further,
using $M_1x = \mu M_2 x$, the cost at the optimum is given by
$\lambda \mu$ and the maximum and minimum values of the cost given in
the lemma follow.
\end{proof}

\blue{\begin{lemma} \label{lem:H_til_pinv}
Let $\tilde{H} = H^{\trans} \Sigma^{-1} H$ where $\Sigma>0$ and $H$ has a full row rank. Then, $\tilde{H}^{+} = H^{+}\Sigma (H^{+})^{\trans}$.
\end{lemma}
\begin{proof}
Let $\Sigma = RR^{\trans}$ be the Cholesky decomposition.
\begin{align*}
\tilde{H}^{+}  &= (   (R^{-1}H)^{\trans} R^{-1}H  )^{+} = (R^{-1}H)^{+} ((R^{-1}H)^{+})^{\trans} \\
& \overset{(a)}{=} H^{+} R R^{\trans}(H^{+})^{\trans} = H^{+}\Sigma (H^{+})^{\trans},
\end{align*}
where $(a)$ follows since $H$ is full row rank. 
\end{proof}
}


\end{document}